\documentclass[11pt]{article}

\usepackage[T1]{fontenc}
\usepackage[utf8]{inputenc}
\usepackage{amsmath,amssymb,mathtools,amsthm}
\usepackage{booktabs,tabularx,array,enumitem,longtable}
\usepackage{microtype}
\usepackage[margin=1in]{geometry}
\usepackage[hidelinks,hypertexnames=false]{hyperref}
\usepackage{algorithm,algpseudocode}
\usepackage[nameinlink,noabbrev]{cleveref}

\setlist{leftmargin=2.1em,itemsep=0.15em,topsep=0.2em}

\hypersetup{
  pdftitle={Fairness--Stability Trade-offs in Many-to-One Matching},
  pdfauthor={Genjie Qin}
}

\newtheorem{theorem}{Theorem}
\newtheorem{lemma}{Lemma}
\newtheorem{proposition}{Proposition}
\newtheorem{corollary}{Corollary}
\theoremstyle{definition}
\newtheorem{definition}{Definition}
\newtheorem{remark}{Remark}

\newcommand{\F}{F}
\newcommand{\W}{W}
\newcommand{\SW}{\operatorname{SW}}
\newcommand{\OPT}{\operatorname{OPT}}
\newcommand{\OPTr}{\operatorname{OPT}_{r}}

\newcommand{\EFXp}{\mathrm{EFX}^{+}}
\newcommand{\pos}[1]{\left[#1\right]_+}
\newcommand{\argminop}{\operatorname*{arg\,min}}
\newcommand{\argmaxop}{\operatorname*{arg\,max}}
\newcommand{\Top}{\operatorname{Top}}

\title{Fairness--Stability Trade-offs in Many-to-One Matching}
\author{Genjie Qin\\
Ocean University of China\\
\texttt{qgj@stu.ouc.edu.cn}}
\date{}

\begin{document}
\maketitle

\begin{abstract}
We study the trade-off between firm-side fairness and coalition stability in many-to-one matching markets with transferable payments. For a fixed matching $X$, we characterize the largest supportable core factor by a bottleneck financing problem: $\alpha(X)=1/\Phi(X)$, where $\Phi(X)=\min_{z\ge0}\max_i R_i(X,z)$. This yields a polynomial-time linear program and local sensitivity formulas for one-worker reallocations. We then develop a maximum-edge round algorithm and a broader class of mutual-top safe choices. Every safe execution is EF1 and, with $t=\delta(A)$ denoting the minimum positive-edge quality, guarantees $\alpha(X)\ge\max\{t,1/[m-(m-1)t]\}$ and $\SW(X)/\OPT\ge t+(1-t)/m$. These bounds give finite-firm lower and upper bounds for the EF1--core minimax frontier, with exact results for two firms and for three firms when $\delta\le1/2$; as the number of firms grows, the tight scale-free stability rate is $\delta$. We also extend the financing formulation to stronger $\EFXp$ fairness and capacity-constrained markets.
\end{abstract}

\section{Introduction}\label{sec:intro}

\subsection{Research question: when fairness moves the market away from the exact core}
Talent, project, and institutional allocation problems often involve efficiency, fairness, and stability at the same time. Examples include internal labor markets, team formation on platforms, hospital staff allocation, student--project assignment, and admissions with capacity constraints. The operations research literature has studied stable matching through mathematical programming, polyhedral methods, and algorithms~\cite{RothRothblumVandeVate1993,SethuramanTeoQian2006,DelormeEtAl2019,PerachAnily2022}. A related literature studies tractable trade-offs between fairness and efficiency in resource allocation~\cite{ArgyrisKarsuYavuz2022,AzizEtAl2023,DemeulemeesterEtAl2025}. We focus on a setting in which these requirements directly conflict: how much coalition stability must be lost when a planner requires a more balanced allocation of workers across firms or teams?

Our baseline model is a many-to-one allocation market with transferable payments. Wages may be interpreted as actual compensation, or more broadly as internal transfer prices or shadow compensation used by a central platform. Firm-side EF1 compares worker bundles from the firms' own valuation perspectives. It is therefore a bundle-balance requirement across firms, not a substitute for worker-side welfare, opportunity fairness, or justified envy. This interpretation leads to a natural OR design problem: first choose a personnel allocation, then redesign compensation so that the resulting allocation is as stable as possible.

The Gale--Shapley model provides the basic framework for many-to-one stable matching~\cite{GaleShapley1962}. With transferable payments, the Shapley--Shubik assignment game connects welfare-maximizing assignments, competitive prices, and the core through linear programming duality~\cite{ShapleyShubik1971}. Recent work on many-to-one assignment markets studies the structure of the core and competitive salaries~\cite{AtayNunezSolymosi2024}, while matching with transfers under distributional constraints shows that LP methods remain useful when additional allocation constraints are present~\cite{JalotaOstrovskyPavone2022}. On the fairness side, EF1 is one of the standard relaxations of envy-freeness for indivisible goods~\cite{LiptonEtAl2004,CaragiannisEtAl2019,BarmanEtAl2018,AmanatidisEtAl2023}, and fair-division ideas have increasingly been studied in two-sided and many-to-one matching environments~\cite{FreemanMichaShah2021,IgarashiEtAl2024,IgarashiEtAl2025,NarangBiswasNarahari2022}.

The two objectives need not be compatible. An exact-core outcome must attain unconstrained optimal welfare, so each worker must be assigned to a firm with maximum value for that worker. Firm-side EF1 may instead require moving some high-value workers away from firms that receive too concentrated a bundle. Fairness therefore need not be a selection rule within the stable set; it can move the allocation away from the welfare-optimal configuration that supports the exact core.

We do not ask whether exact core and EF1 can always be achieved simultaneously. Instead, we study their quantitative compatibility. Once fairness moves the matching away from the welfare optimum, how much stability is lost? Can the loss be computed exactly for a fixed matching? And, in the worst case, how large a core factor can the most stable EF1 matching still support? We use the path
\begin{equation}\label{eq:path}
X^0\;(\mathrm{OPT},\alpha=1)\longrightarrow X^1\longrightarrow\cdots\longrightarrow X^T\;(\mathrm{EF1})
\end{equation}
as an analytical device rather than as a commitment to a specific local algorithm. It leads to three questions:
\begin{enumerate}[label=(\roman*),leftmargin=2.4em]
\item For any intermediate matching $X$, can the best core factor $\alpha(X)$ be characterized and computed exactly?
\item How does a one-worker fairness update change the stability constraints, and how can the local effect propagate across firms?
\item For a fixed number of firms and a fixed instance-quality parameter, what is the worst-case frontier between EF1 and approximate core stability?
\end{enumerate}
The third question is the main frontier problem; the first two provide the structural and algorithmic tools needed to answer it.

\subsection{A common tool: stability as a financing ratio}
For worker $j$, let $M_j$ be the highest market value and let $d_j$ be the value created at the current employer. If fairness moves the worker from a highest-value firm to a firm with $d_j<M_j$, keeping the old competitive wage immediately creates an excess wage requirement of $M_j-d_j$ relative to the new output. Lowering the wage may reduce this pressure at the current employer, but it can simultaneously increase the gains available to other firms from recruiting the worker. Stability loss is therefore not simply welfare loss. It is a bottleneck created after current wage requirements and potential deviation gains are rebalanced across firms.

Given normalized wages $z=(z_j)$, define
\begin{equation}\label{eq:intro-G}
G_i(z):=\max_{T\subseteq W}\sum_{j\in T}(a_{ij}-z_j)
=\sum_{j\in W}\pos{a_{ij}-z_j},
\end{equation}
and firm $i$'s total normalized stabilization requirement
\begin{equation}\label{eq:intro-D}
D_i(X,z)=\underbrace{\sum_{j\in X_i}z_j}_{\text{normalized wages of current workers}}+
\underbrace{G_i(z)}_{\text{normalized deviation benchmark}}.
\end{equation}
The set $T$ may include current workers. When $\alpha=1$, $G_i(z)$ is the largest profit firm $i$ can retain after forming a new team while preserving the selected workers' current payoffs. For a general $\alpha$-core, it should be read as a scaled deviation benchmark rather than as an actual realized profit. Let $B_i$ be the output of firm $i$'s current bundle, and define
\begin{equation}\label{eq:intro-R}
R_i(X,z):=\frac{D_i(X,z)}{B_i}.
\end{equation}
This ratio is the firm-level state variable used throughout the paper. Under fixed wages, the system bottleneck is
\begin{equation}\label{eq:intro-rho}
\rho(X,z):=\max_iR_i(X,z),
\end{equation}
and after wage re-optimization,
\begin{equation}\label{eq:intro-Phi}
\Phi(X)=\min_{z\ge0}\rho(X,z)=\min_{z\ge0}\max_iR_i(X,z).
\end{equation}
Our central representation is
\begin{equation}\label{eq:intro-id}
\alpha(X)=\frac{1}{\Phi(X)}.
\end{equation}
Thus, $\Phi(X)$ has a direct operational interpretation: it is the minimum achievable worst-firm stabilization pressure after compensation is optimally redesigned for a fixed personnel allocation. The ``pressure'' is not a literal subsidy budget; it is a normalized measure implied by coalition-stability constraints. A smaller $\Phi$ means that the allocation is easier to stabilize, or equivalently that it supports a larger core factor. This also explains why $\Phi$ is used later as a secondary criterion among fairness-safe personnel allocations: once fairness and worst-case guarantees are protected, $\Phi$ measures the remaining difficulty of coalition stabilization.

The identity itself is not the end of the analysis. It converts coalition stability into a minimax problem over firm bottlenecks. Locally, we can track how a fairness update changes each $R_i$. Globally, we do not require $\Phi$ to improve monotonically along a repair path. Instead, we construct a fair terminal matching for which all $R_i$ admit a common upper bound. The maximum-edge round algorithm achieves this by controlling direct value loss, the branching of deviation pressure, and bottleneck congestion.

\subsection{Main results and contributions}
We organize the theoretical contribution in four layers.

\textbf{1. Exact stability of a fixed matching and local sensitivity.}
Starting from the coalition-core constraints, we show that single-firm deviation inequalities are necessary and sufficient. This gives the identity in~\eqref{eq:intro-id} and an explicit linear program. Hence, even when fairness moves the allocation away from the welfare optimum, we can compute the largest core factor and supporting wages exactly. For a one-worker move $j:k\to\ell$, we track the firm ratios $R_i$ directly. Under fixed wages, only $R_k$ and $R_\ell$ change, so one can determine whether the current bottleneck falls, stays unchanged, or rises, and whether wage adjustment is needed at all. Only after this fixed-wage test do we consider wage rebalancing. Lowering $z_j$ is one useful local direction: it reduces pressure at the current employer but may raise pressure at several competing firms. We derive closed-form ratio changes, an exact re-optimization identity for $\Phi(X')$, and a one-step certificate based on the old optimal wages.

\textbf{2. Safe-round robustness and when choice really matters.}
The sensitivity analysis highlights three risks: direct value loss $M_j-d_j$, branching when a wage reduction raises several competing firms' $R_i$, and congestion when this pressure reaches firms with small output. We show that the maximum-edge round guarantee requires only two local dominance conditions: the selected worker is row-top for the firm, and the selected firm is column-top for the worker among active firms. These mutual-top safe edges always exist and preserve exact EF1 and the same pressure bounds. Any two safe edges with disjoint endpoints commute. A convenient ex-ante sufficient condition is that each firm has a strict ranking over its positive-value workers and each worker has a strict ranking over positive-value firms. Under this condition, all safe-round executions produce the same positive-value worker matching. Genuine personnel choices can arise only from candidates that share a firm or a worker. We therefore use $\Phi$ only as a secondary stability criterion for such conflicts. Whether or not the optional conflict resolver is used, the safe algorithm family guarantees exact EF1 and
\begin{equation*}
\alpha(X)\ge \Gamma_m(\delta(A))
=\max\left\{\delta(A),\frac{1}{m-(m-1)\delta(A)}\right\},
\end{equation*}
with welfare ratio at least $u_m(\delta(A))=\delta(A)+(1-\delta(A))/m$.

\textbf{3. The EF1--core minimax frontier.}
For $m$ firms and positive-edge quality at least $\delta$, define
\begin{equation*}
H_m(\delta)=\inf_{|\F|=m,\,\delta(A)\ge\delta}\max_{X:\,\mathrm{EF1}}\alpha(X).
\end{equation*}
We prove
\begin{equation*}
\Gamma_m(\delta)\le H_m(\delta)\le u_m(\delta).
\end{equation*}
The bounds close in several important cases: $H_m(0)=1/m$; for two firms, $H_2(\delta)=(1+\delta)/2$ for the full range; for three firms, $H_3(\delta)=(1+2\delta)/3$ when $\delta\le1/2$, while the largest remaining gap for $\delta>1/2$ is below $0.01197$. For general finite $m\ge4$, we give explicit lower and upper bounds without claiming a complete characterization. When the number of firms is unbounded, the tight scale-free stability rate under any fixed positive fraction of EF1 is exactly $\delta$. Thus, $\delta$ is not only an analytical parameter; it is the exact cross-scale limit of fairness--stability compatibility.

\textbf{4. Robustness of the financing framework.}
We extend the framework to stronger $\EFXp$ fairness and to capacity constraints. While preserving the scale-free optimal $\delta$-core benchmark, we guarantee at least $\max\{\delta,1/2\}$-$\EFXp$. With capacities, a firm can select at most $r_i$ workers in a deviation, so the deviation benchmark becomes the sum of the largest at most $r_i$ positive terms. This gives a capacity analogue $\alpha_r=1/\Phi_r$ and corresponding EF1--core guarantees. These extensions are used to test whether the same financing logic survives changes in the fairness notion and the deviation set.

\subsection{Relation to the literature}\label{sec:related}

\paragraph{Stable matching, assignment games, and operations research.}
Gale--Shapley stable matching provides the basic structure for college admissions and many-to-one matching~\cite{GaleShapley1962}. From a linear-programming perspective, Roth, Rothblum, and Vande Vate link stable matchings, optimal assignments, and polyhedral structure~\cite{RothRothblumVandeVate1993}, while Sethuraman, Teo, and Qian study geometry and fairness in many-to-one stable matching~\cite{SethuramanTeoQian2006}. Related EJOR work develops discrete-optimization models for stable matching with ties and incomplete lists and for stable assignment of student groups to dormitories~\cite{DelormeEtAl2019,PerachAnily2022}. These papers mainly treat stability as a feasibility requirement or optimization objective. We instead allow fairness to move the market away from complete stability and optimize how much stability can still be retained.

\paragraph{Transferable payments, the core, and allocation constraints.}
The Shapley--Shubik assignment game uses dual prices to characterize the core of one-to-one transferable-utility markets~\cite{ShapleyShubik1971}. Recent work on many-to-one assignment markets continues to study the core and competitive salaries~\cite{AtayNunezSolymosi2024}. Jalota, Ostrovsky, and Pavone study many-to-one markets with transfers under distributional constraints and use LP duality to characterize equilibrium existence and computation~\cite{JalotaOstrovskyPavone2022}. Conflicts between stability and allocation constraints have also been addressed through integer programming in project allocation~\cite{AgostonBiroSzanto2018}. In our setting, the unconstrained market itself has an exact core. Approximate stability is created endogenously by firm-side fairness.

\paragraph{Fair allocation of indivisible goods and efficiency trade-offs.}
EF1 originates from approximate envy-freeness for indivisible goods~\cite{LiptonEtAl2004}, while systematic work on EFX includes Plaut and Roughgarden~\cite{PlautRoughgarden2017}; a recent survey reviews the algorithmic development of EF1, EFX, MMS, and related notions~\cite{AmanatidisEtAl2023}. Under additive valuations, the relation between maximum Nash welfare, EF1, and Pareto efficiency shows that fairness and efficiency can coexist in important settings~\cite{CaragiannisEtAl2019,BarmanEtAl2018}. Budish studies indivisibility and fairness through approximate competitive equilibrium~\cite{Budish2011}. A particularly relevant OR viewpoint is to optimize welfare subject to fairness constraints or to select fairly among optimal solutions. EJOR has studied welfare-dominance constraints for fair allocation, welfare-maximizing EF1 allocation, and fairness in general integer programming~\cite{ArgyrisKarsuYavuz2022,AzizEtAl2023,DemeulemeesterEtAl2025}. We carry this ``optimization under fairness'' logic from welfare objectives to coalition stability.

\paragraph{Fairness in two-sided and many-to-one matching.}
Freeman, Micha, and Shah bring fair-division ideas into two-sided many-to-many matching~\cite{FreemanMichaShah2021}. Igarashi et al. study fair division with two-sided preferences and, more recently, two-sided fairness in many-to-one matching~\cite{IgarashiEtAl2024,IgarashiEtAl2025}. Narang, Biswas, and Narahari study leximin fairness within stable many-to-one matchings~\cite{NarangBiswasNarahari2022}. These papers typically impose stability exactly, or use pairwise or justified-envy notions of stability. Our direction is reversed: firm-side EF1 is the main fairness requirement, and we choose the most stable outcome within the fair set.

\paragraph{Approximate core and the distinction of our setting.}
Approximate core has also been studied in matching games and multiple-partners matching games~\cite{Vazirani2021,XiaoLuFang2021}. The source of approximation is different here. We do not approximate because the base core is empty or difficult to compute. We approximate because EF1 can force the matching away from the allocation that supports the exact core. This leads us to study three linked objects: the best core factor of a fixed non-optimal matching, the local sensitivity of stability to a fairness move, and the cross-instance frontier
\[
H_m(\delta)=\inf_{|\F|=m,\,\delta(A)\ge\delta}\max_{X:\,\mathrm{EF1}}\alpha(X).
\]

\paragraph{Contribution and OR positioning.}
We do not treat the algebraic normalization $\alpha(X)=1/\Phi(X)$ as the sole novelty. Its value is that it supports the following local-to-global chain:
\begin{equation*}
\begin{aligned}
&\text{fixed-matching stability representation}
\Rightarrow \text{local sensitivity}\\
&\Rightarrow \text{safe-round robustness}
\Rightarrow \text{commutativity and ex-ante uniqueness}\\
&\Rightarrow \text{conflict-triggered secondary stability optimization}\\
&\Rightarrow \text{minimax frontier}.
\end{aligned}
\end{equation*}
The structure separates personnel allocation from compensation design. Safe rounds first protect EF1 and worst-case core/welfare guarantees. Only endpoint conflicts require a further comparison of personnel alternatives, and $\Phi$ measures the remaining stabilization difficulty after compensation is optimally redesigned. The paper therefore combines structural bounds with tractable LPs, local diagnostics, and an optional conflict-resolution rule, in line with the OR view of discrete allocation under fairness constraints.

\section{Model, stability, and fairness benchmarks}\label{sec:model}

Let $\F=\{1,\ldots,m\}$ be the set of firms and $\W=\{1,\ldots,n\}$ the set of workers. Firm $i$ values worker $j$ at $a_{ij}\ge0$. Values are additive, so $V_i(T)=\sum_{j\in T}a_{ij}$. A complete matching $X=(X_1,\ldots,X_m)$ is a partition of $\W$, and $\mu(j)$ denotes worker $j$'s current employer. Define
\begin{equation}\label{eq:basic-values}
M_j=\max_i a_{ij},\qquad d_j=a_{\mu(j)j},\qquad B_i=V_i(X_i).
\end{equation}
Social welfare and unconstrained optimal welfare are
\begin{equation}\label{eq:welfare}
\SW(X)=\sum_iB_i=\sum_jd_j,\qquad \OPT=\sum_jM_j.
\end{equation}
If at least one positive-value edge exists, define the positive-edge quality
\begin{equation}\label{eq:delta}
\delta(A)=\min_{a_{ij}>0}\frac{a_{ij}}{M_j};
\end{equation}
for the all-zero market, set $\delta(A)=1$.

Firm profits $x_i\ge0$ and worker wages $y_j\ge0$ satisfy local financing if
\begin{equation}\label{eq:local-finance}
x_i+\sum_{j\in X_i}y_j=B_i\qquad\forall i.
\end{equation}
For a coalition of firms $I\subseteq\F$ and a set of workers $T\subseteq\W$, define the deviation value
\begin{equation}\label{eq:coalition-value}
v(I,T)=\sum_{j\in T}\max_{i\in I}a_{ij}.
\end{equation}
An outcome $(X,x,y)$ is in the $\alpha$-core if, for all $I,T$,
\begin{equation}\label{eq:alpha-core-def}
x(I)+y(T)\ge\alpha v(I,T).
\end{equation}
For a fixed matching $X$, let $\alpha(X)$ denote the largest $\alpha$ that can be supported by nonnegative payments satisfying local financing.

\begin{remark}[Worker set in a single-firm deviation]\label{rem:deviation-set}
For a single firm $i$, the deviation set $T\subseteq\W$ is chosen from the entire worker set; we do \emph{not} require $T\cap X_i=\varnothing$. The firm may keep some current workers, release others, and recruit workers from other firms. Therefore, the positive-part expression $\sum_j\pos{\alpha a_{ij}-y_j}$ is exactly the maximum over all $T\subseteq\W$. Only after we use the wage-clipping result in Proposition~\ref{prop:clipping}, with $d_j\le z_j\le M_j$, do current workers cease to make a positive contribution to $\pos{a_{ij}-z_j}$ in the unconstrained model. This is a consequence of clipping, not a restriction on the deviation set.
\end{remark}

\begin{definition}[EF1, $\beta$-EF1, and $\gamma$-$\EFXp$]\label{def:fairness}
A matching $X$ is EF1 if, for every pair of firms $i,k$ with $X_k\neq\varnothing$, there exists $g\in X_k$ such that
\begin{equation*}
V_i(X_i)\ge V_i(X_k\setminus\{g\}).
\end{equation*}
For $\beta\in(0,1]$, $X$ is $\beta$-EF1 if, for every $i,k$ with $X_k\neq\varnothing$, there exists $g\in X_k$ such that
\begin{equation*}
V_i(X_i)\ge \beta V_i(X_k\setminus\{g\}).
\end{equation*}
Thus, $1$-EF1 is EF1. For $\gamma\in[0,1]$, $X$ is $\gamma$-$\EFXp$ if, for every $i,k$ and every $g\in X_k$ with $a_{ig}>0$,
\begin{equation*}
V_i(X_i)\ge\gamma V_i(X_k\setminus\{g\}).
\end{equation*}
When $\gamma=1$, we simply write $\EFXp$.
\end{definition}

\begin{remark}[Fairness interpretation]\label{rem:fairness-interpretation}
The EF1/$\EFXp$ notion used here is \emph{firm-side bundle fairness}: firm $i$ compares its own worker bundle with other firms' bundles using its own valuations. This is natural for a central platform or organization in which several teams compete for a common talent pool. It is not the same as worker-side individual fairness or justified envy. Stability is measured by a transferable-utility coalition core. We therefore study a quantitative trade-off between two different design objectives rather than using EF1 as a replacement for standard two-sided stability.
\end{remark}

\begin{lemma}[Single-firm deviation condition]\label{lem:deficit}
Under local financing, $(X,x,y)$ belongs to the $\alpha$-core if and only if, for every firm $i$,
\begin{equation}\label{eq:deficit}
x_i\ge\max_{T\subseteq\W}\sum_{j\in T}(\alpha a_{ij}-y_j)
=\sum_{j\in\W}\pos{\alpha a_{ij}-y_j}.
\end{equation}
The set $T$ may contain current workers of firm $i$. When $\alpha=1$, the middle expression is the largest profit firm $i$ can retain after selecting a new team and preserving the selected workers' current payoffs. For a general $\alpha$-core, it is an $\alpha$-scaled deviation benchmark.
\end{lemma}
\begin{proof}
Necessity follows by considering the singleton coalition $I=\{i\}$. For a fixed firm $i$, the core condition is
\begin{equation*}
x_i\ge\sum_{j\in T}(\alpha a_{ij}-y_j)\qquad\forall T\subseteq\W.
\end{equation*}
The right-hand side is separable across workers, so its maximum is attained by selecting all positive terms. This gives~\eqref{eq:deficit}.

Conversely, suppose~\eqref{eq:deficit} holds for every firm. Consider any coalition $I,T$. For each $j\in T$, choose a firm attaining $\max_{i\in I}a_{ij}$ and partition the workers into disjoint sets $(T_i)_{i\in I}$ according to these choices. Then
\begin{align*}
\alpha v(I,T)-y(T)
&=\sum_{i\in I}\sum_{j\in T_i}(\alpha a_{ij}-y_j)\\
&\le\sum_{i\in I}\sum_{j\in\W}\pos{\alpha a_{ij}-y_j}\\
&\le\sum_{i\in I}x_i=x(I).
\end{align*}
Rearranging gives~\eqref{eq:alpha-core-def}. Hence the single-firm conditions are equivalent to all coalition constraints.
\end{proof}

\begin{proposition}[Exact core and welfare upper bound]\label{prop:exact}
Every fixed matching satisfies $\alpha(X)\le\SW(X)/\OPT$. Moreover, $X$ supports the exact core if and only if $d_j=M_j$ for every worker $j$.
\end{proposition}
\begin{proof}
Local financing gives $x(\F)+y(\W)=\SW(X)$. Applying the $\alpha$-core condition to the grand coalition $(\F,\W)$ yields $\SW(X)\ge\alpha\OPT$, so $\alpha\le\SW(X)/\OPT$.

If $X$ supports the exact core, then $\SW(X)\ge\OPT$. Since $d_j\le M_j$ worker by worker, we also have $\SW(X)=\sum_jd_j\le\sum_jM_j=\OPT$. Equality of the sums and the termwise inequalities imply $d_j=M_j$ for every $j$.

Conversely, if $d_j=M_j$ for every $j$, set $y_j=M_j$ and $x_i=0$. Local financing holds, and for every $I,T$,
$v(I,T)\le\sum_{j\in T}M_j=y(T)$. Hence the exact-core inequalities hold.
\end{proof}

\begin{proposition}[Numerical characterization of EF1]\label{prop:ef1-iff}
A matching $X$ is EF1 if and only if, for every ordered pair of firms $i,k$,
\begin{equation}\label{eq:ef1-iff}
B_i\ge\sum_{j\in X_k}a_{ij}-\max_{g\in X_k}a_{ig}.
\end{equation}
\end{proposition}
\begin{proof}
For fixed $i,k$, EF1 requires a worker $g\in X_k$ such that
$B_i=V_i(X_i)\ge V_i(X_k)-a_{ig}$. The right-hand side is minimized by removing the worker in $X_k$ with the highest value to firm $i$. Therefore such a $g$ exists if and only if~\eqref{eq:ef1-iff} holds. Imposing the condition for all ordered pairs gives the result.
\end{proof}

\section{The financing--stability identity: exact stability of a fixed matching}\label{sec:identity}

For normalized wages $z=(z_j)\ge0$, define firm $i$'s normalized deviation benchmark
\begin{equation}\label{eq:G}
G_i(z):=\max_{T\subseteq\W}\sum_{j\in T}(a_{ij}-z_j)
=\sum_{j\in\W}\pos{a_{ij}-z_j}.
\end{equation}
The set $T$ may contain current workers and workers employed elsewhere. If $\alpha=1$ and actual wages are $y=z$, then $G_i(z)$ is the largest profit firm $i$ can retain after forming a new team while preserving the selected workers' current payoffs. For a general $\alpha$-core, it is a normalized deviation benchmark generated by scaling; it should not be read as an actual profit that the firm receives.

Define firm $i$'s total stabilization requirement
\begin{equation}\label{eq:D}
D_i(X,z)=\underbrace{\sum_{j\in X_i}z_j}_{\text{normalized wages of current workers}}
+\underbrace{G_i(z)}_{\text{normalized deviation benchmark}}.
\end{equation}
Both terms must be supported by current output: current output pays current workers, while the residual firm profit must be large enough to cover the deviation benchmark. If $B_i>0$, define the stabilization ratio
\begin{equation}\label{eq:R}
R_i(X,z):=\frac{D_i(X,z)}{B_i}.
\end{equation}
A larger $R_i$ means that the stabilization requirement is tighter relative to the firm's current production base. If $B_i=0<D_i$, set $R_i=+\infty$; if $B_i=D_i=0$, set $R_i=0$.

\begin{theorem}[Financing--stability identity]\label{thm:identity}
A fixed matching $X$ supports an $\alpha$-core if and only if there exists $z\ge0$ such that $\alpha R_i(X,z)\le1$ for all firms. Let
\begin{equation}\label{eq:Phi}
\Phi(X)=\min_{z\ge0}\max_iR_i(X,z).
\end{equation}
Then
\begin{equation}\label{eq:alpha-Phi}
\alpha(X)=\frac{1}{\Phi(X)}.
\end{equation}
\end{theorem}
\begin{proof}
Fix $\alpha>0$ and write actual wages as $y_j=\alpha z_j$. By Lemma~\ref{lem:deficit}, firm $i$'s core condition is equivalent to
\begin{equation*}
\begin{aligned}
x_i
&\ge\max_{T\subseteq\W}\left\{\alpha\sum_{j\in T}a_{ij}-\sum_{j\in T}y_j\right\}\\
&=\alpha\max_{T\subseteq\W}\sum_{j\in T}(a_{ij}-z_j)
=\alpha G_i(z).
\end{aligned}
\end{equation*}
This also clarifies the meaning of $G_i(z)$: it is the normalized deviation benchmark obtained after factoring the $\alpha$ scaling out of deviation value. Local financing gives
$x_i=B_i-\alpha\sum_{j\in X_i}z_j$. Substitution yields
\begin{equation*}
B_i\ge\alpha\left(\sum_{j\in X_i}z_j+\sum_j\pos{a_{ij}-z_j}\right)
=\alpha D_i(X,z).
\end{equation*}
Thus $X$ supports an $\alpha$-core if and only if there exists $z\ge0$ with $D_i(X,z)/B_i\le1/\alpha$ for every $i$.

Let $\rho(X,z)=\max_iR_i(X,z)$. For fixed $z$, the largest supportable core factor is $1/\rho(X,z)$. Optimizing over wages redistributes stabilization pressure across firms so that the largest $R_i$ is as small as possible. Hence the optimal core factor is
$1/\min_z\rho(X,z)=1/\Phi(X)$. The zero-output conventions are exactly consistent with the inequalities $B_i\ge\alpha D_i$.
\end{proof}

\begin{proposition}[Wage clipping]\label{prop:clipping}
When computing $\Phi(X)$, we may restrict without loss of generality to $d_j\le z_j\le M_j$ for every worker $j$.
\end{proposition}
\begin{proof}
If $z_j>M_j$, lower it to $M_j$. The current employer's wage term decreases. For every firm $i$, since $a_{ij}\le M_j$, the positive-gap term is zero both before and after the change. Thus no $D_i$ increases.

If $z_j<d_j=a_{\mu(j)j}$, raise it to $d_j$. For the current employer $k=\mu(j)$, the wage term increases by $d_j-z_j$, while the positive gap for worker $j$ falls from $d_j-z_j$ to zero, so the two changes cancel. For every other firm, a higher wage can only reduce the positive gap. Repeating worker by worker gives an optimal solution satisfying $d_j\le z_j\le M_j$.
\end{proof}

\begin{corollary}[Current workers create no positive deviation term under clipped optimal wages]\label{cor:own-worker-zero}
In the unconstrained model, $\Phi(X)$ admits optimal wages satisfying Proposition~\ref{prop:clipping}. For any firm $i$ and current worker $j\in X_i$,
\begin{equation}\label{eq:own-zero}
\pos{a_{ij}-z_j}=\pos{d_j-z_j}=0.
\end{equation}
Thus, although the maximization in~\eqref{eq:G} formally allows $T$ to include current workers, under a clipped optimal representation they contribute no positive term to $G_i(z)$. Positive contributions come only from non-current workers whose wages are below firm $i$'s valuations.
\end{corollary}
\begin{proof}
If $j\in X_i$, then $a_{ij}=d_j$. Proposition~\ref{prop:clipping} gives $z_j\ge d_j$, so $a_{ij}-z_j\le0$.
\end{proof}

\begin{remark}[Boundary for the capacity extension]
Corollary~\ref{cor:own-worker-zero} relies on the clipping result for the unconstrained model. The capacity model uses a Top-$r_i$ aggregation of deviation gains. We therefore do not assume without proof that current workers make no positive contribution in the capacity case; the Top-$r_i$ expression is always taken over the full worker set $\W$.
\end{remark}

\begin{proposition}[Linear-programming formulation]\label{prop:LP}
$\Phi(X)$ is the optimal value of
\begin{equation}\tag{$\mathsf{F}_X$}\label{eq:FX}
\begin{aligned}
\min\quad &\rho\\
\text{s.t.}\quad
&\sum_{j\in X_i}z_j+\sum_js_{ij}\le\rho B_i &&\forall i,\\
&s_{ij}\ge a_{ij}-z_j &&\forall i,j,\\
&s_{ij}\ge0,\quad z_j\ge0,\quad \rho\ge0.
\end{aligned}
\end{equation}
Hence $\alpha(X)$ and a set of optimal wages can be computed exactly in polynomial time.
\end{proposition}
\begin{proof}
For fixed $z$, an optimal solution sets $s_{ij}=\pos{a_{ij}-z_j}$. The first group of constraints is therefore exactly $R_i(X,z)\le\rho$ for all $i$. The LP chooses wages to minimize the largest firm ratio, so its optimal value is $\Phi(X)$. The number of variables and constraints is polynomial.
\end{proof}

\begin{remark}[Worker moves and wage adjustment in the LP]\label{rem:LP-R-dynamics}
Suppose worker $j$ moves from firm $k$ to firm $\ell$. Only two firm rows in the first group of LP constraints change directly: firm $k$ removes $z_j$ from its current-wage term and its output base becomes $B_k-v_k$, while firm $\ell$ adds $z_j$ and its output base becomes $B_\ell+v_\ell$. Other firm rows are unchanged. The constraints $s_{ih}\ge a_{ih}-z_h$ do not depend on the current employer and therefore do not change under the discrete move. One should first check whether the old wages already satisfy the target bottleneck after the move. If so, no wage adjustment is needed. If further optimization is desired, the full wage vector is re-optimized in the same LP. Lowering $z_j$ is only one analytically useful local direction; when it crosses valuation thresholds, some constraints $s_{ij}\ge a_{ij}-z_j$ can become binding and create a new bottleneck.
\end{remark}

For the detailed three-firm analysis and the upper-bound certificates used later, it is convenient to record a primal--dual formulation equivalent to the financing representation. We use $h_{ij}$ for the linearized gap of firm $i$ on worker $j$ to distinguish it from normalized wages $z_j$.

\begin{proposition}[Primal--dual formulation for a fixed matching]\label{prop:primal-dual}
If $\OPT>0$, then $\alpha(X)$ equals the optimal value of
\begin{equation}\tag{$P_X$}\label{eq:PX-journal}
\begin{aligned}
\max\quad &\alpha\\
\text{s.t.}\quad
&\sum_{j\in X_i}y_j+\sum_{j\in\W}h_{ij}\le B_i &&\forall i,\\
&y_j+h_{ij}\ge\alpha a_{ij} &&\forall i,j,\\
&\alpha\ge0,\quad y_j\ge0,\quad h_{ij}\ge0.
\end{aligned}
\end{equation}
Its dual is
\begin{equation}\tag{$D_X$}\label{eq:DX-journal}
\begin{aligned}
\min\quad &\sum_iB_ip_i\\
\text{s.t.}\quad
&\sum_{i,j}a_{ij}q_{ij}\ge1,\\
&\sum_iq_{ij}\le p_{\mu(j)} &&\forall j,\\
&q_{ij}\le p_i &&\forall i,j,\\
&p_i\ge0,\quad q_{ij}\ge0.
\end{aligned}
\end{equation}
The two programs satisfy strong duality and have common optimal value $\alpha(X)$.
\end{proposition}
\begin{proof}
By Lemma~\ref{lem:deficit}, a fixed matching supports an $\alpha$-core if and only if there exist nonnegative wages and profits satisfying
\begin{equation*}
x_i\ge\sum_j\pos{\alpha a_{ij}-y_j}\qquad\forall i.
\end{equation*}
For each $(i,j)$, introduce $h_{ij}\ge0$ with
$h_{ij}\ge\alpha a_{ij}-y_j$. At an optimal feasible solution, $h_{ij}$ can be reduced to the positive part $\pos{\alpha a_{ij}-y_j}$, so the linearization is exact. Local financing gives
\begin{equation*}
x_i=B_i-\sum_{j\in X_i}y_j.
\end{equation*}
Substituting into $x_i\ge\sum_jh_{ij}$ gives the first constraints of $(P_X)$; the second constraints are the gap linearization. Thus the feasible values of $\alpha$ coincide with those of the original core problem.

For the dual, associate $p_i\ge0$ with each firm budget constraint and $q_{ij}\ge0$ with each covering constraint $y_j+h_{ij}\ge\alpha a_{ij}$. Variable $y_j$ appears only in the budget constraint of its current employer $\mu(j)$ and in all covering constraints for worker $j$, giving
\begin{equation*}
\sum_iq_{ij}\le p_{\mu(j)}.
\end{equation*}
Variable $h_{ij}$ gives $q_{ij}\le p_i$. Since $\alpha$ has coefficient $a_{ij}$ in the covering constraints and objective coefficient $1$, we obtain
\begin{equation*}
\sum_{i,j}a_{ij}q_{ij}\ge1.
\end{equation*}
The dual objective is $\sum_iB_ip_i$, which gives $(D_X)$. The primal is feasible and bounded when $\OPT>0$, so standard LP strong duality applies.
\end{proof}

Homogenizing $(D_X)$ gives a form that is useful for structural analysis.

\begin{corollary}[Homogeneous blocking pressure and a worker-wise fractional knapsack]\label{cor:homogeneous-pressure}
Define
\begin{equation}\label{eq:R-journal}
\begin{aligned}
R_A(X)=\max\quad &\sum_{i,j}a_{ij}q_{ij}\\
\text{s.t.}\quad
&\sum_iB_ip_i\le1,\\
&\sum_iq_{ij}\le p_{\mu(j)} &&\forall j,\\
&q_{ij}\le p_i &&\forall i,j,\\
&p_i,q_{ij}\ge0.
\end{aligned}
\end{equation}
Then $\alpha(X)=1/R_A(X)$. Moreover, for fixed firm prices $p$, the maximization over $q$ separates completely by worker. If $j\in X_k$, worker $j$ contributes
\begin{equation}\label{eq:worker-knapsack}
\phi_j(p)=\max\left\{\sum_i a_{ij}q_i:\ 0\le q_i\le p_i,\ \sum_iq_i\le p_k\right\}.
\end{equation}
Hence
\begin{equation}\label{eq:R-worker}
R_A(X)=\max_{\substack{p\ge0\\\sum_iB_ip_i\le1}}\sum_j\phi_j(p),
\end{equation}
and~\eqref{eq:worker-knapsack} is a fractional knapsack problem with total capacity $p_k$, item capacity $p_i$, and unit value $a_{ij}$.
\end{corollary}
\begin{proof}
First consider the reciprocal relation. If $(p,q)$ is feasible for $(D_X)$, let $c=\sum_iB_ip_i>0$. Dividing all $p,q$ by $c$ gives a feasible solution of~\eqref{eq:R-journal} with objective at least $1/c$. Applying this to an optimal solution of $(D_X)$ with $c=\alpha(X)$ gives $R_A(X)\ge1/\alpha(X)$.

Conversely, if $(p,q)$ is feasible for~\eqref{eq:R-journal} with objective $r>0$, divide all variables by $r$. The normalization becomes $\sum_{i,j}a_{ij}q_{ij}=1$, yielding a feasible solution of $(D_X)$ with objective at most $1/r$. At $r=R_A(X)$, Proposition~\ref{prop:primal-dual} gives $\alpha(X)\le1/R_A(X)$. Combining the two inequalities yields $\alpha(X)=1/R_A(X)$.

Now fix $p$. The constraints $q_{ij}\le p_i$ limit worker $j$'s share assigned to firm $i$, while $\sum_iq_{ij}\le p_{\mu(j)}$ couples only the different firm shares of the same worker. There are no shared $q$-constraints across workers, so the objective separates by worker. For $j\in X_k$, we obtain~\eqref{eq:worker-knapsack}. Maximizing over $p$ with $\sum_iB_ip_i\le1$ gives~\eqref{eq:R-worker}. The fractional-knapsack interpretation follows directly from the capacities and unit values.
\end{proof}

\section{Local sensitivity of a fairness adjustment}\label{sec:sensitivity}

Let $X'$ be obtained from $X$ by moving a single worker $j$ from firm $k$ to firm $\ell$. Write
\begin{equation}\label{eq:move-values}
v_k=a_{kj},\qquad v_\ell=a_{\ell j}.
\end{equation}
The updated outputs are $B_k'=B_k-v_k$ and $B_\ell'=B_\ell+v_\ell$, while all other firms keep the same output. We separate a fairness adjustment into three layers. First, we keep wages fixed and isolate the matching effect on the firm ratios $R_i$. Second, we check whether the new fixed-wage bottleneck already meets the target stability level. Only if further improvement is desired do we re-optimize the full wage vector for the new matching. The branching result below studies one useful, but not necessary, local direction $z_j\downarrow$.

\begin{theorem}[Firm-bottleneck sensitivity to a one-worker move]\label{thm:sensitivity}
For any fixed $z\ge0$,
\begin{equation}\label{eq:D-update}
\begin{aligned}
D_i(X',z)&=D_i(X,z), &&i\notin\{k,\ell\},\\
D_k(X',z)&=D_k(X,z)-z_j,\\
D_\ell(X',z)&=D_\ell(X,z)+z_j.
\end{aligned}
\end{equation}
Hence
\begin{equation}\label{eq:R-update}
\begin{aligned}
R_i(X',z)&=R_i(X,z), &&i\notin\{k,\ell\},\\
R_k(X',z)&=\frac{D_k(X,z)-z_j}{B_k-v_k},\\
R_\ell(X',z)&=\frac{D_\ell(X,z)+z_j}{B_\ell+v_\ell}.
\end{aligned}
\end{equation}
Moreover, the optimal post-move bottleneck satisfies the exact identity
\begin{equation}\label{eq:Phi-update}
\begin{aligned}
\Phi(X')=\min_{z\ge0}\max\Bigg\{&
\max_{i\notin\{k,\ell\}}R_i(X,z),\\
&\frac{D_k(X,z)-z_j}{B_k-v_k},\\
&\frac{D_\ell(X,z)+z_j}{B_\ell+v_\ell}
\Bigg\}.
\end{aligned}
\end{equation}
\end{theorem}
\begin{proof}
The deviation term $G_i(z)=\sum_h\pos{a_{ih}-z_h}$ depends only on valuations and wages, not on the current employer of each worker. It is therefore unchanged when the matching moves from $X$ to $X'$. In the current-wage term, only $z_j$ is removed from firm $k$ and added to firm $\ell$, which gives~\eqref{eq:D-update}. Substituting the updated outputs $B_k-v_k$ and $B_\ell+v_\ell$ into $R_i=D_i/B_i$ gives~\eqref{eq:R-update}. Finally, minimizing the largest updated ratio over all wages gives~\eqref{eq:Phi-update}.
\end{proof}

\begin{corollary}[Closed-form fixed-wage changes and sign tests]\label{cor:ratio-diff}
If the relevant denominators are positive, then
\begin{equation}\label{eq:ratio-diff}
\begin{aligned}
R_k(X',z)-R_k(X,z)
&=\frac{v_kR_k(X,z)-z_j}{B_k-v_k},\\
R_\ell(X',z)-R_\ell(X,z)
&=\frac{z_j-v_\ell R_\ell(X,z)}{B_\ell+v_\ell}.
\end{aligned}
\end{equation}
Therefore,
\begin{equation}\label{eq:ratio-sign}
\begin{aligned}
R_k(X',z)<R_k(X,z)
&\iff \frac{z_j}{v_k}>R_k(X,z),\\
R_\ell(X',z)>R_\ell(X,z)
&\iff \frac{z_j}{v_\ell}>R_\ell(X,z),
\end{aligned}
\end{equation}
whenever $v_k>0$ or $v_\ell>0$ for the corresponding comparison.
\end{corollary}
\begin{proof}
Use~\eqref{eq:R-update}, bring the terms to a common denominator, and substitute $D_i=B_iR_i$:
\begin{align*}
\frac{D_k-z_j}{B_k-v_k}-R_k
&=\frac{v_kR_k-z_j}{B_k-v_k},\\
\frac{D_\ell+z_j}{B_\ell+v_\ell}-R_\ell
&=\frac{z_j-v_\ell R_\ell}{B_\ell+v_\ell}.
\end{align*}
With positive denominators, the signs are determined by the numerators.
\end{proof}

Equation~\eqref{eq:ratio-sign} gives a direct interpretation of the discrete matching effect. If worker $j$ is ``expensive relative to output'' for the source firm, in the sense that $z_j/v_k>R_k$, moving the worker out reduces the source firm's stabilization pressure. If the worker is expensive relative to productivity at the receiving firm, $z_j/v_\ell>R_\ell$, receiving the worker raises that firm's pressure. A firm that is ``rich'' or ``poor'' in the EF1 sense need not be loose or tight in the stability sense. Hence a one-step fairness repair has no general monotone effect on system stability.

\begin{corollary}[Fixed-wage bottleneck test]\label{cor:fixed-wage-test}
For any fixed wage vector $z$, the post-move fixed-wage bottleneck is
\begin{equation}\label{eq:fixed-wage-rho-after-move}
\rho(X',z)
=
\max\left\{
\frac{D_k(X,z)-z_j}{B_k-v_k},
\frac{D_\ell(X,z)+z_j}{B_\ell+v_\ell},
\max_{i\notin\{k,\ell\}}R_i(X,z)
\right\}.
\end{equation}
Thus the matching move alone may yield
\begin{equation}\label{eq:rho-three-directions}
\rho(X',z)<\rho(X,z),\qquad
\rho(X',z)=\rho(X,z),\qquad
\rho(X',z)>\rho(X,z).
\end{equation}
For a target stability rate $\bar\alpha$, if
\begin{equation}\label{eq:no-wage-needed}
\rho(X',z)\le\frac1{\bar\alpha},
\end{equation}
then the same actual wages $y=\bar\alpha z$ already support a $\bar\alpha$-core for $X'$, so no wage adjustment is needed.
\end{corollary}
\begin{proof}
Equation~\eqref{eq:fixed-wage-rho-after-move} follows by taking the maximum of the updated ratios in Theorem~\ref{thm:sensitivity}. All three directions are possible because the source and receiving firms have different local sign conditions in~\eqref{eq:ratio-diff}. If~\eqref{eq:no-wage-needed} holds, then $\bar\alpha R_i(X',z)\le1$ for all firms; Theorem~\ref{thm:identity} then gives a feasible $\bar\alpha$-core at the same normalized wages.
\end{proof}

\begin{remark}[The matching effect and the wage effect must be separated]\label{rem:matching-wage-separation}
Corollary~\ref{cor:fixed-wage-test} shows that moving a worker does not automatically require a wage cut or any wage adjustment. The matching move first reallocates pressure between the source and receiving firms under the old wages. Only when the resulting bottleneck remains too high is there a reason to re-optimize wages. Even then, the mathematical problem is
\begin{equation*}
\Phi(X')=\min_{z\ge0}\max_iR_i(X',z),
\end{equation*}
so the full wage vector may change. The direction $z_j\downarrow$ considered below is analytically important, but it is not a necessary form of the optimal wage update.
\end{remark}

\begin{corollary}[One-step stability certificate]\label{cor:one-step}
Let $z^*$ be an optimal normalized wage vector for $X$. Define
\begin{equation}\label{eq:hatPhi}
\begin{aligned}
\widehat\Phi_{j:k\to\ell}(z^*)=\max\Bigg\{&
\max_{i\notin\{k,\ell\}}R_i(X,z^*),\\
&\frac{D_k(X,z^*)-z_j^*}{B_k-v_k},\\
&\frac{D_\ell(X,z^*)+z_j^*}{B_\ell+v_\ell}
\Bigg\}.
\end{aligned}
\end{equation}
Then $\Phi(X')\le\widehat\Phi_{j:k\to\ell}(z^*)$, and hence
\begin{equation}\label{eq:one-step-alpha}
\alpha(X')\ge\frac{1}{\widehat\Phi_{j:k\to\ell}(z^*)}.
\end{equation}
\end{corollary}
\begin{proof}
Equation~\eqref{eq:Phi-update} minimizes over all $z\ge0$. Evaluating the new matching at the old optimal wages $z^*$ gives
$\widehat\Phi_{j:k\to\ell}(z^*)=\rho(X',z^*)$ and therefore
$\Phi(X')\le\widehat\Phi_{j:k\to\ell}(z^*)$. The certificate measures the largest firm bottleneck after changing the matching while temporarily keeping wages fixed. It does not assume that $z_j$ is lowered. Taking reciprocals and using Theorem~\ref{thm:identity} gives~\eqref{eq:one-step-alpha}.
\end{proof}

\paragraph{A local pre-screening rule.}
Given a candidate set $\mathcal C(X)$ generated by a fairness rule, first compute an optimal wage vector $z^*$ for the current matching, and then choose
\begin{equation}\label{eq:oracle}
(j,k,\ell)^*\in\argminop_{(j,k,\ell)\in\mathcal C(X)}
\widehat\Phi_{j:k\to\ell}(z^*).
\end{equation}
The rule selects the feasible fairness move with the smallest next-step worst-firm certificate under the old optimal wages. If an exact local comparison is desired, one can solve~\eqref{eq:Phi-update} for a small set of the best-scoring candidates. This is a local diagnostic tool. We do not claim that it produces a globally optimal terminal matching, makes $\Phi$ monotone, or guarantees finite convergence for arbitrary fairness-repair dynamics.

\begin{proposition}[The first move from an exact-core benchmark]\label{prop:first-step}
Suppose $X$ is welfare maximizing and worker $j$ moves from a highest-value firm to firm $\ell$, where the post-move value is $d_j<M_j$. Let the pre-move output of firm $\ell$ be $B_\ell$ and assume $B_\ell+d_j>0$. Under the old high-wage certificate $z=M$,
\begin{equation}\label{eq:first-step}
R_k(X',M)=1,\qquad
R_\ell(X',M)=\frac{B_\ell+M_j}{B_\ell+d_j}>1,
\end{equation}
so
\begin{equation}\label{eq:first-step-certificate}
\Phi(X')\le\frac{B_\ell+M_j}{B_\ell+d_j},\qquad
\alpha(X')\ge\frac{B_\ell+d_j}{B_\ell+M_j}.
\end{equation}
\end{proposition}
\begin{proof}
At a welfare-maximizing matching, set normalized wages $z_h=M_h$. Since $M_h\ge a_{ih}$ for every firm, $G_i(M)=0$. Current values equal highest values, so $D_i(X,M)=B_i$ and $R_i(X,M)=1$ for every firm. After the move, the source firm loses worker $j$'s output $M_j$ and the same wage requirement, so its ratio remains $1$. The receiving firm's wage requirement rises by $M_j$, while output rises only by $d_j$, giving $(B_\ell+M_j)/(B_\ell+d_j)$. All other ratios remain $1$. Evaluating $\Phi(X')$ at these wages gives the certificate. The difference between numerator and denominator is exactly $M_j-d_j$, the direct source of the new high-wage bottleneck.
\end{proof}

\begin{remark}[First-step loss versus the certificate]\label{rem:first-step-strict}
Because $d_j<M_j$, Proposition~\ref{prop:exact} implies that $X'$ cannot support the exact core. Hence $\alpha(X')<1$ and $\Phi(X')>1$. Equation~\eqref{eq:first-step-certificate} is an upper-bound certificate obtained by keeping the old high-wage scheme; it does not claim equality after wage re-optimization. Thus the first nontrivial fairness move away from an exact-core allocation must raise the optimal bottleneck above $1$, but later repair steps need not be monotone.
\end{remark}

If the fixed-wage test already meets the target stability rate, no wage adjustment is needed. If the planner wants to reduce the bottleneck further, the entire wage vector can be re-optimized. The next result studies one especially transparent direction: lowering the wage of a transferred worker when its current employer is tight.

\begin{proposition}[Firm-bottleneck branching under a local wage reduction]\label{prop:branching}
Fix the matching and all wages other than $z_j$. Suppose worker $j$ is currently employed by firm $\ell$, and consider $z_j\mapsto z_j-\varepsilon$ within the clipping region, where $\varepsilon>0$ and $z_j-\varepsilon\ge d_j$. Assume that no valuation threshold is crossed in this interval, and define the active competing-firm set
\begin{equation}\label{eq:active-set}
A_j(z):=\{i\neq\ell:a_{ij}>z_j\}.
\end{equation}
Then
\begin{equation}\label{eq:branching-R}
\begin{aligned}
R_\ell(X,z-\varepsilon e_j)
&=R_\ell(X,z)-\frac{\varepsilon}{B_\ell},\\
R_i(X,z-\varepsilon e_j)
&=R_i(X,z)+\frac{\varepsilon}{B_i},
&&i\in A_j(z),\\
R_h(X,z-\varepsilon e_j)
&=R_h(X,z),
&&h\notin A_j(z)\cup\{\ell\}.
\end{aligned}
\end{equation}
Thus a wage reduction relaxes one bottleneck at the current employer but may tighten several competing-firm bottlenecks. For a competing firm with small $B_i$, the same wage change increases $R_i$ more sharply; this is the firm-level congestion effect.
\end{proposition}
\begin{proof}
By Corollary~\ref{cor:own-worker-zero}, within the clipping region the current employer $\ell$ has no positive deviation term from its own worker $j$. Lowering the wage therefore reduces $D_\ell$ by $\varepsilon$. For each $i\in A_j(z)$, worker $j$ contributes $a_{ij}-z_j$ to $G_i(z)$, so a wage reduction increases $D_i$ by $\varepsilon$. Inactive firms are unchanged. Since the matching is fixed, the output bases $B_i$ do not change. Dividing by $B_i$ yields~\eqref{eq:branching-R}. The raw total $\sum_iD_i$ changes by $(|A_j(z)|-1)\varepsilon$, but this total is not the stability objective; stability is governed by the largest $R_i$.
\end{proof}

Under fixed wages, define the worst-firm ratio
\begin{equation}\label{eq:rho-fixed}
\rho(X,z):=\max_iR_i(X,z).
\end{equation}
From~\eqref{eq:branching-R}, within the local interval,
\begin{equation}\label{eq:rho-wage-update}
\begin{aligned}
\rho(X,z-\varepsilon e_j)
=\max\Bigg\{&
R_\ell(X,z)-\frac{\varepsilon}{B_\ell},\\
&\max_{i\in A_j(z)}\left(R_i(X,z)+\frac{\varepsilon}{B_i}\right),\\
&\max_{h\notin A_j(z)\cup\{\ell\}}R_h(X,z)
\Bigg\}.
\end{aligned}
\end{equation}
Along this optional direction, the fixed-wage bottleneck may fall, remain unchanged, or rise. If the current employer is the unique bottleneck and competing firms have enough slack, a moderate wage reduction can lower $\rho$. If a competitor is already close to the bottleneck, the same reduction can make it the new worst firm. This does not imply that the optimal wage vector must move in this direction. The optimal stability of the fixed matching remains
\begin{equation*}
\Phi(X)=\min_z\rho(X,z).
\end{equation*}
Hence, except for the first move from an exact-core benchmark in Proposition~\ref{prop:first-step}, a general fairness-repair path need not satisfy either $\Phi(X^{t+1})\ge\Phi(X^t)$ or the reverse inequality.

The mechanism of this section can be summarized as follows. A fairness repair first changes the source and receiving firms' $R_i$ through the matching transfer under old wages. This matching effect can improve, preserve, or worsen the bottleneck. If the old wages are already sufficient, the process stops there. Wage re-optimization may then redistribute stability pressure further across firms. The direction $z_j\downarrow$ is one important local example, and the branching result describes how it can relax the current employer while tightening several competitors. Final approximate stability is determined by the globally re-optimized value $\Phi$.

\section[Design principle: from local Ri dynamics to a global bottleneck bound]{Design principle: from local $R_i$ dynamics to a global bottleneck bound}\label{sec:design}

Section~\ref{sec:sensitivity} shows that the basic object in a fairness repair is the full vector of firm bottlenecks
\begin{equation*}
(R_1,\ldots,R_m),
\end{equation*}
not a single aggregate deficit. A discrete worker move directly changes only the source and receiving firms' ratios, and this fixed-wage matching effect has no predetermined direction. If the old wages already meet the target, there is no need to change them. Only after wage re-optimization can pressure propagate to other firms. The main role of local sensitivity is therefore to identify which matching moves and possible wage responses make the final value
\begin{equation*}
\max_iR_i
\end{equation*}
hard to control. It is not used to prove that $\Phi$ improves monotonically along a fairness-repair path.

This distinction matters. An exact-core starting point has $\Phi=1$, which is the smallest possible value. Once fairness moves a worker to a firm with strictly lower value than the market maximum, the new matching has $\Phi>1$. Later repairs may alleviate an existing bottleneck or create a new one. The maximum-edge round algorithm therefore cannot be interpreted as an algorithm that decreases $\Phi$ at every step.

From the $R_i$ perspective, a global fair construction must control three risks:
\begin{equation}\label{eq:three-design-risks}
\begin{array}{rcl}
\text{direct value loss} &:& M_j-d_j,\\
\text{bottleneck branching} &:& \text{number of competing }R_i\text{ raised by one wage reduction},\\
\text{bottleneck congestion} &:& \displaystyle \max_i R_i(X,z).
\end{array}
\end{equation}
The first determines how large a new bottleneck can appear under a high-wage scheme after a worker is moved away from a highest-value firm. The second is a conditional wage-rebalancing risk: if the planner lowers a wage to relax the current employer, how many competing firms become tighter? The third scales these changes by each firm's output base $B_i$, explaining why the same deviation pressure is more serious when concentrated at low-output firms.

The maximum-edge round algorithm in Section~\ref{sec:round} does not solve the local monotonicity problem. It solves a more useful worst-case problem:
\begin{quote}
\textbf{Construct an exact-EF1 terminal matching $X$ and exhibit explicit wage schemes under which every firm ratio $R_i(X,z)$ is bounded by a common constant.}
\end{quote}
It uses three combinatorial controls:
\begin{enumerate}[label=(\roman*),leftmargin=2.4em]
\item \textbf{Value retention:} workers are assigned only along positive edges, so every nonzero worker satisfies $d_j\ge\delta(A)M_j$, controlling direct high-wage mismatch.
\item \textbf{Branching control:} each firm receives at most one worker per round, so after a firm exits a round it faces at most $m-1$ later competitive exposures in that round.
\item \textbf{Congestion control:} the maximum-edge order relates these exposures to the firm's own production base $B_i$, yielding a firm-wise bound on relative rematching pressure.
\end{enumerate}

The local-to-global logic is therefore
\begin{equation}\label{eq:local-to-global-chain}
\begin{aligned}
&\text{local sensitivity: track how each repair changes }(R_1,\ldots,R_m),\\
&\text{global construction: do not require path monotonicity; bound all terminal }R_i,\\
&\text{minimax frontier: compare the constructive guarantee with the best possible EF1 stability.}
\end{aligned}
\end{equation}
Section~\ref{sec:round} also shows that the maximum-edge rule can be relaxed to mutual-top safe edges without losing the general finite-$m$ guarantees. More importantly, we give an ex-ante uniqueness condition. If every firm and every worker has strict local rankings on positive edges, then simultaneously available safe edges have disjoint endpoints and commute, so the terminal matching is independent of safe-choice order. Equal edge values are not a problem by themselves; only row or column ties that create shared-endpoint conflicts can change the personnel allocation. In those cases, an optional stability-aware resolver compares complete candidate matchings under a common benchmark completion using exact $\Phi$, or uses $\widehat\Phi$ as a lower-cost certificate. The pure maximum-edge algorithm remains the simplest default rule and is also used in the sharper three-firm frontier analysis.

\section{Low-dimensional structure: the two- and three-firm breakpoints}\label{sec:smallm}

\begin{theorem}[Welfare--stability identity for two firms]\label{thm:two}
If $m=2$ and $\OPT>0$, then every complete matching $X$ satisfies
\begin{equation}\label{eq:two}
\alpha(X)=\frac{\SW(X)}{\OPT}.
\end{equation}
\end{theorem}
\begin{proof}
The upper bound follows from Proposition~\ref{prop:exact}. It remains to construct payments that achieve the welfare ratio. Let $B_i=V_i(X_i)$, $B=B_1+B_2$, and $\rho=B/\OPT$. Define the positive rematching gains
\begin{equation*}
G_1=\sum_{j\in X_2}\pos{a_{1j}-a_{2j}},\qquad
G_2=\sum_{j\in X_1}\pos{a_{2j}-a_{1j}}.
\end{equation*}
For each worker, the difference between the highest value and the current value is exactly the positive gain from the other firm. Hence
\begin{equation*}
\OPT=B+G_1+G_2.
\end{equation*}
Multiplying by $\rho=B/\OPT$ and rearranging gives
\begin{equation}\label{eq:two-balance}
(1-\rho)B=\rho(G_1+G_2).
\end{equation}
Start with proportional wages $y_j^0=\rho a_{\mu(j)j}$ and profits $x_i^0=(1-\rho)B_i$. Firm $1$'s total positive gap is $\rho G_1$, and firm $2$'s is $\rho G_2$. By~\eqref{eq:two-balance}, total firm profit equals total gap.

If $x_i^0\ge\rho G_i$ for both firms, Lemma~\ref{lem:deficit} already gives a $\rho$-core. Otherwise, exactly one firm has a profit surplus and the other has an equal deficit. Suppose firm $1$ has surplus $s=x_1^0-\rho G_1>0$. Firm $2$ then has an equal shortfall $s$. Its gap comes only from workers in $X_1$ with $a_{2j}>a_{1j}$. Increase these workers' wages gradually by a total amount $s$. Each unit of wage increase reduces firm $1$'s profit by one unit and reduces firm $2$'s gap on that worker by one unit. Stop when firm $1$'s surplus is exhausted. Because firm $2$'s total shortfall is exactly $s$, the increases can be chosen within the original gaps to remove the shortfall exactly. Higher wages do not create a gap for firm $1$ on its own workers, and with only two firms no third-party gap can appear. The adjusted profits cover both firms' gaps, so a $\rho$-core is supported. Therefore $\alpha(X)\ge\rho$, and the welfare upper bound gives equality.
\end{proof}

The special feature of two firms is that a worker has at most one competitor other than the current employer. If global wage re-optimization moves locally in a direction $z_j\downarrow$, at most one current-employer ratio falls and one competing-firm ratio rises. With three firms, the same wage direction may raise two external firms' ratios at once. Thus one relaxed bottleneck can correspond to several tightened bottlenecks. This explains why wage rebalancing can create stability pressure beyond welfare loss when $m\ge3$. Branching is the local mechanism behind this difficulty; it is not a required step after every matching move. The next example makes the separation strict.

\begin{proposition}[A strict welfare--stability gap with three firms]\label{prop:three-gap}
Let
\begin{equation}\label{eq:three-matrix}
A=\begin{pmatrix}
1&99/100&79/100\\
99/100&4/5&79/100\\
99/100&99/125&16/25
\end{pmatrix},\qquad X_i=\{i\}.
\end{equation}
Then
\begin{equation}\label{eq:three-gap}
\alpha(X)=\frac{44}{51}<\frac{\SW(X)}{\OPT}=\frac{122}{139}.
\end{equation}
\end{proposition}
\begin{proof}
Take $\alpha=44/51$ and wages
\begin{equation*}
(y_1,y_2,y_3)=\left(\frac{363}{425},\frac{967}{1275},\frac{16}{25}\right).
\end{equation*}
Local financing gives profits $x_1=62/425$, $x_2=53/1275$, and $x_3=0$. Direct calculation of $\pos{\alpha a_{ij}-y_j}$ shows that firm $1$'s total positive gap is exactly $62/425$, firm $2$'s is $53/1275$, and firm $3$ has no positive gap. Lemma~\ref{lem:deficit} therefore gives $\alpha(X)\ge44/51$.

For the upper bound, use a dual certificate for the fixed-matching LP. Set
\begin{equation*}
p=\left(\frac{100}{357},\frac{100}{357},\frac{200}{357}\right)
\end{equation*}
and let $q_{11}=q_{12}=q_{13}=q_{23}=100/357$, with all other $q_{ij}=0$. These variables satisfy the worker-capacity and firm upper-bound constraints in the dual, and the normalized blocking value is $1$. The dual objective is $44/51$. Strong duality gives $\alpha(X)\le44/51$, so equality holds. Finally, $\SW(X)=61/25$ and $\OPT=139/50$, yielding $\SW/\OPT=122/139$.
\end{proof}

\section{Global construction: from local financing risks to EF1--core guarantees}\label{sec:round}

Section~\ref{sec:design} identified the design objective. The goal is not to decrease $\Phi$ step by step along a repair path. Instead, we seek a simple construction whose terminal matching admits wage schemes that bound every firm's ratio $R_i$. The maximum-edge round algorithm uses only valuation rankings but produces two complementary bottleneck certificates. The first keeps wages high, eliminates all deviation benchmarks, and controls direct value loss. The second evaluates wages at current productivity and uses the round structure to control each firm's external rematching pressure $P_i/B_i$. These two certificates give firm-wise bounds on $R_i$ and together yield $\Gamma_m(\delta)$.

\begin{algorithm}[ht]
\caption{Global maximum-edge round algorithm}\label{alg:round}
\begin{algorithmic}[1]
\State Set $X_i\gets\varnothing$ for all firms and let $U\gets\W$ be the unassigned workers
\While{there exist $i\in\F$ and $j\in U$ with $a_{ij}>0$}
  \State Start a new round and set the active-firm set $S\gets\F$
  \While{there exist $i\in S$ and $j\in U$ with $a_{ij}>0$}
    \State Choose a maximum-value edge $(i^*,j^*)$ among all such edges
    \State Set $X_{i^*}\gets X_{i^*}\cup\{j^*\}$ and $U\gets U\setminus\{j^*\}$
    \State Set $S\gets S\setminus\{i^*\}$; firm $i^*$ leaves the current round
  \EndWhile
\EndWhile
\State Assign all-zero workers arbitrarily
\end{algorithmic}
\end{algorithm}

\begin{lemma}[Common first-round removal witness]\label{lem:common-witness}
Suppose firm $k$ receives positive-value workers $p_k^1,p_k^2,\ldots$ across rounds. If $p_k^1$ exists, then for every firm $i$,
\begin{equation}\label{eq:common-witness}
V_i(X_i)\ge V_i(X_k\setminus\{p_k^1\}).
\end{equation}
Hence Algorithm~\ref{alg:round} outputs an exact-EF1 matching.
\end{lemma}
\begin{proof}
Let $p_i^t$ denote the worker received by firm $i$ in round $t$, if such a positive-value worker exists. Consider worker $p_k^{t+1}$. When firm $i$ chooses $p_i^t$ in round $t$, worker $p_k^{t+1}$ is still unassigned because it is not assigned to $k$ until the next round; firm $i$ is also still active immediately before its own choice. The maximum-edge rule therefore gives
\begin{equation*}
a_{i,p_i^t}\ge a_{i,p_k^{t+1}}.
\end{equation*}
If $p_i^t$ does not exist, then firm $i$ remains active until the end of the round and has no positive edge to any remaining worker. Hence $a_{i,p_k^{t+1}}=0$ for every worker assigned only in a later round, and the same inequality still holds. Summing over $t\ge1$, the left-hand side is at most $V_i(X_i)$, while the right-hand side covers all positive-value workers in $X_k$ except $p_k^1$. All-zero workers have zero value to every firm. Thus~\eqref{eq:common-witness} holds and provides a common EF1 removal witness.
\end{proof}

\begin{lemma}[Value retention and a direct stability certificate]\label{lem:positive-edge}
Suppose every worker with $M_j>0$ is assigned along a positive-value edge in matching $X$, and $\delta(A)\ge t$. Then $d_j\ge tM_j$ for every worker, so $M_j-d_j\le(1-t)M_j$. Moreover, $X$ supports a $t$-core and satisfies $\SW(X)/\OPT\ge t$.
\end{lemma}
\begin{proof}
Set each worker's wage to $y_j=tM_j$. Since every current edge is positive, the definition of $\delta(A)$ gives $d_j\ge tM_j$. Define firm profits by
\begin{equation*}
x_i=B_i-t\sum_{j\in X_i}M_j\ge0.
\end{equation*}
For any coalition $I,T$, we have $v(I,T)\le\sum_{j\in T}M_j$. Hence
\begin{equation*}
x(I)+y(T)\ge y(T)=t\sum_{j\in T}M_j\ge t\,v(I,T),
\end{equation*}
which gives a $t$-core. Summing $d_j\ge tM_j$ over workers yields $\SW(X)\ge t\OPT$.
\end{proof}

To strengthen the finite-$m$ stability guarantee using the round structure, define firm $i$'s total positive rematching pressure by
\begin{equation}\label{eq:Pi}
P_i(X)=\sum_{j\in\W}\pos{a_{ij}-d_j},\qquad
\kappa(X)=\max_{i:B_i>0}\frac{P_i(X)}{B_i}.
\end{equation}
For zero-output firms, we use the same infinite-ratio convention as before.

\begin{lemma}[Proportional-wage bridge]\label{lem:kappa-bridge}
Every fixed matching $X$ supports a $1/(1+\kappa(X))$-core.
\end{lemma}
\begin{proof}
Let $\alpha=1/(1+\kappa)$, set $y_j=\alpha d_j$, and set $x_i=(1-\alpha)B_i$. Firm $i$'s total positive gap is
\begin{equation*}
\sum_j\pos{\alpha a_{ij}-\alpha d_j}=\alpha P_i(X)\le\alpha\kappa B_i.
\end{equation*}
Since $1-\alpha=\alpha\kappa$, the profit exactly covers the gap. Lemma~\ref{lem:deficit} gives the claim. Equivalently, in the financing representation one may set $z_j=d_j$ and obtain $D_i(X,d)=B_i+P_i(X)\le(1+\kappa)B_i$, hence $\Phi(X)\le1+\kappa$.
\end{proof}

\begin{lemma}[Incoming pressure under the round algorithm]\label{lem:round-pressure}
Let $X$ be the output of Algorithm~\ref{alg:round} and let $t=\delta(A)$. Then for every firm $i$,
\begin{equation}\label{eq:round-pressure}
P_i(X)\le(m-1)(1-t)B_i.
\end{equation}
\end{lemma}
\begin{proof}
Fix firm $i$ and analyze positive gains entering $i$ round by round. If $i$ does not receive a positive-value worker in a round, it remains active throughout that round. Whenever another firm selects worker $j$, the chosen edge is a maximum positive edge among active firms and unassigned workers, so the current employer value satisfies $d_j\ge a_{ij}$. Thus worker $j$ creates no positive rematching gain for firm $i$.

Now suppose $i$ receives a worker of value $b$ in a round and then exits. Workers selected before $i$ exits again create no positive gain for $i$. After $i$ exits, at most $m-1$ other firms can each select one worker in that round. Any such worker $j$ was still unassigned when $i$ selected its own worker, so the maximum-edge rule gives $a_{ij}\le b$. Since the selected edge for $j$ is positive, $d_j\ge tM_j\ge t a_{ij}$. Hence
\begin{equation*}
\pos{a_{ij}-d_j}\le(1-t)a_{ij}\le(1-t)b.
\end{equation*}
The round contributes at most $(m-1)(1-t)b$ to $P_i$. Summing over the rounds in which firm $i$ receives a worker, and noting that the corresponding $b$ values sum to $B_i$, proves~\eqref{eq:round-pressure}.
\end{proof}

\begin{theorem}[Two certificates from local bottlenecks to global stability]\label{thm:local-global-bridge}
Let $X$ be the output of Algorithm~\ref{alg:round} and let $t=\delta(A)>0$. There are two explicit normalized wage vectors under which every firm satisfies
\begin{equation}\label{eq:two-funding-certificates}
\begin{aligned}
R_i(X,M)&\le \frac{1}{t} &&\forall i,\\
R_i(X,d)&\le m-(m-1)t &&\forall i.
\end{aligned}
\end{equation}
Consequently,
\begin{equation}\label{eq:Phi-round-bridge}
\Phi(X)\le
\min\left\{\frac1t,\,m-(m-1)t\right\},
\qquad
\alpha(X)\ge
\max\left\{t,\frac1{m-(m-1)t}\right\}.
\end{equation}
The first certificate controls direct mismatch after workers leave highest-value edges. The second uses the low-wage benchmark $z=d$ to control the external rematching pressure $P_i/B_i$. This pressure is related to the local branching and congestion mechanisms, but the certificate does not require wages to move gradually from $M$ to $d$.
\end{theorem}
\begin{proof}
First set $z_j=M_j$. Since $M_j\ge a_{ij}$ for every firm, $G_i(M)=0$ and
\begin{equation*}
D_i(X,M)=\sum_{j\in X_i}M_j.
\end{equation*}
The algorithm assigns workers only along positive edges, so $d_j\ge tM_j$ for every $j\in X_i$. Hence
\begin{equation*}
R_i(X,M)
=
\frac{\sum_{j\in X_i}M_j}{B_i}
\le\frac1t
\qquad\forall i.
\end{equation*}
This is the high-wage extreme: deviation benchmarks are zero, and the only issue is whether the firm's output can support the market wage requirements of its current workers. The direct mismatch $M_j-d_j$ is the worker-level source of this pressure.

Next set $z_j=d_j$. Current normalized wages then equal current output,
\begin{equation*}
\sum_{j\in X_i}d_j=B_i.
\end{equation*}
Therefore
\begin{equation}\label{eq:R-pressure}
R_i(X,d)
=
\frac{B_i+P_i(X)}{B_i}
=
1+\frac{P_i(X)}{B_i}.
\end{equation}
By Lemma~\ref{lem:round-pressure},
\begin{equation*}
P_i(X)\le(m-1)(1-t)B_i,
\end{equation*}
so
\begin{equation*}
R_i(X,d)\le1+(m-1)(1-t)=m-(m-1)t
\qquad\forall i.
\end{equation*}
This is a direct evaluation at the low-wage benchmark $z=d$: current wage requirements equal production, and all additional stability pressure appears as external rematching pressure $P_i$. The algorithmic round structure controls $P_i/B_i$; it does not require an actual wage path from $M$ to $d$.

Since $\Phi(X)$ re-optimizes over all wages, it is no larger than the worst-firm ratio under either explicit wage vector. Taking the better of the two certificates yields the bound on $\Phi(X)$, and $\alpha(X)=1/\Phi(X)$ gives~\eqref{eq:Phi-round-bridge}.
\end{proof}

Equation~\eqref{eq:Phi-round-bridge} gives a direct bottleneck interpretation of the function $\Gamma_m(t)$:
\begin{equation}\label{eq:Gamma-interpretation}
\Gamma_m(t)=
\max\left\{
\underbrace{t}_{\text{direct-mismatch certificate at high wages}},\,
\underbrace{\frac1{m-(m-1)t}}_{\text{propagated-pressure certificate at low wages}}
\right\}.
\end{equation}
The algorithm does not make $\Phi$ decrease during construction. It guarantees that, at the terminal matching, at least two explicit wage schemes bound all firm ratios at the stated levels. The optimal wage vector may rebalance pressure further and give an even larger stability rate.

\begin{theorem}[Joint guarantee of EF1 and approximate stability]\label{thm:round-guarantee}
Algorithm~\ref{alg:round} outputs an exact-EF1 matching $X$. Let $t=\delta(A)$ and define
\begin{equation}\label{eq:Gamma-u}
\Gamma_m(t)=\max\left\{t,\frac{1}{m-(m-1)t}\right\},\qquad
u_m(t)=t+\frac{1-t}{m}.
\end{equation}
Then
\begin{equation}\label{eq:round-main}
\alpha(X)\ge\Gamma_m(t),\qquad \frac{\SW(X)}{\OPT}\ge u_m(t).
\end{equation}
\end{theorem}
\begin{proof}
Exact EF1 follows from Lemma~\ref{lem:common-witness}. The stability guarantee follows directly from Theorem~\ref{thm:local-global-bridge}: the high-wage certificate controls direct mismatch, while the low-wage certificate controls external rematching pressure $P_i/B_i$. Taking the better certificate gives $\alpha(X)\ge\Gamma_m(t)$.

For welfare, group assignments by round. Suppose a round assigns $r\le m$ positive-value workers. Let their market maxima be $M_1,\ldots,M_r$, where worker $1$ is chosen first. The first edge is a global maximum positive edge among the remaining workers, so its current value equals $M_1$, and $M_1\ge M_s$ for all other workers in the round. Every other selected positive edge has current value at least $t$ times the corresponding market maximum. Hence round welfare is at least
\begin{align*}
M_1+t\sum_{s=2}^rM_s
&=t\sum_{s=1}^rM_s+(1-t)M_1\\
&\ge\left(t+\frac{1-t}{r}\right)\sum_{s=1}^rM_s\\
&\ge\left(t+\frac{1-t}{m}\right)\sum_{s=1}^rM_s.
\end{align*}
Summing over rounds gives $\SW(X)\ge u_m(t)\OPT$. All-zero workers have no effect.
\end{proof}

\subsection{Safe-round robustness, ex-ante uniqueness, and conflict-triggered secondary stability resolution}\label{sec:safe-confluence}

The maximum-edge round algorithm gives a simple worst-case construction. We now identify the local structure that actually supports its guarantees and answer a more practical question: when does a safe choice really matter? The answer has four parts. First, maximum-edge selection can be relaxed to mutual-top safe edges without weakening the general finite-$m$ guarantees. Second, safe candidates with disjoint endpoints commute and differ only in execution order. Third, strict row and column rankings provide an ex-ante uniqueness condition; unrelated equal edge values are harmless. Fourth, only candidates that share a firm or a worker can create a genuine personnel choice, and only then do we use $\Phi$ as a secondary stability objective.

Given the current uncommitted worker set $U$ and active firm set $S$ within a round, define the mutual-top safe-edge set
\begin{equation}\label{eq:safe-edge-set}
\mathcal S(U,S)=
\left\{
(i,j)\in S\times U:
\begin{array}{l}
a_{ij}>0,\\
a_{ij}=\max_{h\in U}a_{ih},\\
a_{ij}=\max_{r\in S}a_{rj}
\end{array}
\right\}.
\end{equation}
Also define the set of current global maximum positive edges
\begin{equation}\label{eq:global-edge-set}
\mathcal G(U,S)
=
\argmaxop_{\substack{r\in S,\ h\in U\\a_{rh}>0}}a_{rh}.
\end{equation}

\begin{lemma}[Existence of a safe edge]\label{lem:safe-existence}
Whenever $S\times U$ contains a positive edge,
\begin{equation}\label{eq:global-safe-inclusion}
\mathcal G(U,S)\subseteq\mathcal S(U,S),
\end{equation}
and therefore $\mathcal S(U,S)\neq\varnothing$.
\end{lemma}
\begin{proof}
Take $(i,j)\in\mathcal G(U,S)$. Since it attains the maximum positive value over all of $S\times U$, it is in particular a maximum among firm $i$'s edges to workers in $U$, and a maximum among active firms' edges to worker $j$. Thus it satisfies both the row-top and active-column-top conditions in~\eqref{eq:safe-edge-set}.
\end{proof}

It is useful to state the relaxed rule as a separate algorithm. Algorithm~\ref{alg:safe-round-optional} has two layers. The safe set determines which edges are admissible. An optional stability resolver chooses among safe edges only when a genuine endpoint conflict occurs. If the secondary objective is not used, any admissible safe edge may be selected.

\begin{algorithm}[ht]
\caption{Safe-round algorithm with an optional stability conflict resolver}\label{alg:safe-round-optional}
\begin{algorithmic}[1]
\State Set $X_i\gets\varnothing$ for all firms and let $U\gets\W$ be the uncommitted workers
\While{there exist $i\in\F$ and $j\in U$ with $a_{ij}>0$}
  \State Start a new round; set $S\gets\F$ and $\mathsf{first}\gets\text{true}$
  \While{there exist $i\in S$ and $j\in U$ with $a_{ij}>0$}
    \If{$\mathsf{first}$}
      \State $\mathcal C_t\gets\mathcal G(U,S)$
    \Else
      \State $\mathcal C_t\gets\mathcal S(U,S)$
    \EndIf
    \State Check whether candidates in $\mathcal C_t$ share a firm or a worker
    \If{the secondary stability objective is enabled and a nontrivial conflict component $\mathcal K_t$ exists}
      \State $e_t\gets\Call{StabilityConflictResolver}{X,U,\mathcal K_t}$
    \Else
      \State Choose any $e_t=(i,j)\in\mathcal C_t$
    \EndIf
    \State Commit $(i,j)$: $X_i\gets X_i\cup\{j\}$, $U\gets U\setminus\{j\}$, $S\gets S\setminus\{i\}$
    \State $\mathsf{first}\gets\text{false}$
  \EndWhile
\EndWhile
\State Assign all-zero workers by any fixed rule
\end{algorithmic}
\end{algorithm}

\begin{theorem}[Robustness of mutual-top safe rounds]\label{thm:safe-round-family}
Every execution of Algorithm~\ref{alg:safe-round-optional}, with or without the secondary stability resolver, outputs an exact-EF1 matching $X$. Let $t=\delta(A)$. Then
\begin{equation}\label{eq:safe-pressure}
P_i(X)\le(m-1)(1-t)B_i\qquad\forall i.
\end{equation}
Consequently,
\begin{equation}\label{eq:safe-core}
\alpha(X)\ge
\Gamma_m(t)
=
\max\left\{
t,\frac1{m-(m-1)t}
\right\}.
\end{equation}
Because the first edge of every round is still a current global maximum positive edge, we also have
\begin{equation}\label{eq:safe-welfare}
\frac{\SW(X)}{\OPT}
\ge
u_m(t)
=
t+\frac{1-t}{m}.
\end{equation}
\end{theorem}
\begin{proof}
We first prove EF1. Index rounds globally and let $p_i^q$ be the positive-value worker received by firm $i$ in round $q$, if one exists. Once a firm remains active to the end of a round without receiving a positive-value worker, there is no positive edge from that firm to any remaining worker, so it will never receive a positive-value worker in a later round. Hence a firm's positive-value workers, if any, occur in a consecutive block of rounds starting from the first round.

If firm $k$ has a worker $p_k^{q+1}$ in round $q+1$, then $p_k^{q+1}$ is still in $U$ when firm $i$ selects $p_i^q$ in round $q$. The row-top condition gives
\begin{equation*}
a_{i,p_i^q}\ge a_{i,p_k^{q+1}}.
\end{equation*}
If $i$ receives no positive-value worker in a round, it has zero value for every worker committed only later. Summing over rounds yields
\begin{equation*}
V_i(X_i)\ge V_i(X_k\setminus\{p_k^1\}),
\end{equation*}
so the output is exact EF1.

Now fix firm $i$ and prove~\eqref{eq:safe-pressure}. Before firm $i$ exits a round, if another active firm $r$ selects worker $j$, the active-column-top condition gives
\begin{equation*}
d_j=a_{rj}\ge a_{ij},
\end{equation*}
so worker $j$ contributes zero to $P_i$. Suppose $i$ receives worker $p$ at value $b=a_{ip}>0$ and exits. For any worker $j$ selected later in the same round, the row-top condition gives $a_{ij}\le b$. Since the selected edge for $j$ is positive and $t=\delta(A)$,
\begin{equation*}
d_j\ge tM_j\ge t a_{ij},
\end{equation*}
and therefore
\begin{equation*}
\pos{a_{ij}-d_j}\le(1-t)a_{ij}\le(1-t)b.
\end{equation*}
At most $m-1$ other firms can commit one worker each after $i$ exits the round. Thus this round contributes at most $(m-1)(1-t)b$ to $P_i$. Summing over rounds and using that the corresponding $b$ values sum to $B_i$ gives~\eqref{eq:safe-pressure}.

All workers with $M_j>0$ are committed along positive edges, so $d_j\ge tM_j$. Evaluating at $z=M$ gives $R_i(X,M)\le1/t$, while evaluating at $z=d$ and using~\eqref{eq:safe-pressure} gives
\begin{equation*}
R_i(X,d)=1+\frac{P_i(X)}{B_i}\le m-(m-1)t.
\end{equation*}
Taking the better certificate gives~\eqref{eq:safe-core}. Finally, the first edge of each round belongs to $\mathcal G(U,S)$ with $S=\F$, so its current value equals the worker's market maximum and is at least the market maximum of any later worker selected in the same round. All later positive edges satisfy $d_j\ge tM_j$. The welfare calculation from Theorem~\ref{thm:round-guarantee} therefore applies unchanged, proving~\eqref{eq:safe-welfare}.
\end{proof}

Theorem~\ref{thm:safe-round-family} separates feasibility from choice. Any safe-edge execution preserves exact EF1, the $\Gamma_m(t)$ core guarantee, and the $u_m(t)$ welfare guarantee. We now ask a secondary decision question: when do different safe candidates represent only different execution orders, and when can they lead to different personnel allocations?

\begin{lemma}[Disjoint safe choices commute]\label{lem:safe-disjoint-commute}
At a safe-round decision state, let the admissible candidate set be
\begin{equation}\label{eq:admissible-candidate-set}
\mathcal C(U,S)=
\begin{cases}
\mathcal G(U,S), & \text{at the first step of a round},\\
\mathcal S(U,S), & \text{at later steps of the round}.
\end{cases}
\end{equation}
Suppose
\begin{equation*}
e=(i,j),\qquad f=(r,h)\in\mathcal C(U,S)
\end{equation*}
have disjoint endpoints,
\begin{equation}\label{eq:disjoint-safe}
i\neq r,\qquad j\neq h.
\end{equation}
Then choosing $e$ followed by $f$, or $f$ followed by $e$, commits the same two edges and reaches the same next round state.
\end{lemma}
\begin{proof}
If the current step is not the first step of the round, then $e,f\in\mathcal S(U,S)$. Choosing $e$ removes only firm $i$ and worker $j$. Since $r\neq i$ and $h\neq j$, both endpoints of $f$ remain. Removing another firm can only reduce worker $h$'s active column competitors, and removing another worker can only reduce firm $r$'s row competitors. Hence $f$ remains mutual-top. The reverse order is symmetric.

If the current step is the first step of the round, then $e,f\in\mathcal G(U,S)$. After selecting one of them, the round moves to a later step. The other edge still has both endpoints available and was originally a global maximum positive edge. After the first pair of endpoints is removed, it remains row-top and active-column-top, hence belongs to the new $\mathcal S$. In either order, the same two firms leave $S$ and the same two workers leave $U$, so the resulting state is identical.
\end{proof}

This lemma gives a useful rule of thumb: equal edge values do not matter by themselves; endpoint conflicts do. For example, $a_{11}=a_{22}$ creates no genuine choice if candidates $(1,1)$ and $(2,2)$ have disjoint endpoints.

\begin{definition}[Strict local rankings on positive edges]\label{def:strict-row-column}
An instance has \emph{strict row rankings} if, for every firm $i$ and distinct workers $j\neq h$, whenever $a_{ij}>0$ and $a_{ih}>0$,
\begin{equation*}
a_{ij}\neq a_{ih}.
\end{equation*}
It has \emph{strict column rankings} if, for every worker $j$ and distinct firms $i\neq r$, whenever $a_{ij}>0$ and $a_{rj}>0$,
\begin{equation*}
a_{ij}\neq a_{rj}.
\end{equation*}
\end{definition}

\begin{theorem}[An ex-ante uniqueness condition]\label{thm:strict-rank-unique}
If an instance has both strict row and strict column rankings, then all safe-round selection rules produce the same final matching on positive-value workers.
\end{theorem}
\begin{proof}
Fix any reachable state. If two admissible safe edges share the same firm $i$, then both workers must be row-top for $i$ in the current set $U$, so the two positive edge values are equal, contradicting strict row rankings. If two admissible safe edges share the same worker $j$, then both firms must be active-column-top for $j$ in the current set $S$, again implying equal positive edge values and contradicting strict column rankings. Therefore any simultaneously admissible safe edges have disjoint endpoints.

By Lemma~\ref{lem:safe-disjoint-commute}, any two first-step choices from the same reachable state form a local diamond: after both edges are executed, the same state is reached. We then induct on the number of uncommitted positive-value workers. If two complete executions make the same first choice, apply the induction hypothesis to the common next state. If they make different first choices, commute these choices so that both have been executed and then apply the induction hypothesis. Each step permanently commits one positive-value worker, so the induction terminates and the final matching is unique.
\end{proof}

The condition is an ex-ante sufficient condition and is strictly weaker than requiring all positive entries of the full matrix to be pairwise distinct. Unrelated edges may have the same value. For example,
\begin{equation}\label{eq:irrelevant-tie-example}
A=
\begin{pmatrix}
1&0\\
0&1
\end{pmatrix}
\end{equation}
has equal positive edge values, but the two candidates have disjoint endpoints and the final matching is unique. If positive-edge valuations are drawn from a jointly absolutely continuous distribution, strict row and column rankings hold with probability one, and so does uniqueness.

\paragraph{When can safe choices change the final matching?}
Failure of strict local rankings is only an ex-ante warning; it does not imply that different final matchings must occur. The relevant online test is whether current admissible candidates share an endpoint. For two safe candidates
\begin{equation*}
e=(i,j),\qquad f=(r,h),
\end{equation*}
there are three structurally different cases:
\begin{enumerate}[label=(\roman*),leftmargin=2.4em]
\item If $i\neq r$ and $j\neq h$, Lemma~\ref{lem:safe-disjoint-commute} shows that the difference is only execution order; no additional optimization is needed.
\item If $j=h$ and $i\neq r$, the two candidates compete for the same worker. Once the worker's employer is committed, choosing $(i,j)$ and choosing $(r,j)$ necessarily lead to different final matchings for worker $j$.
\item If $i=r$ and $j\neq h$, the same firm can commit only one of the workers before leaving the round. The future feasible path changes, so the terminal matching may differ, although later rounds may also reconverge.
\end{enumerate}
The third case can be seen from two two-firm examples. If
\begin{equation*}
A=
\begin{pmatrix}
10&10\\
0&0
\end{pmatrix},
\end{equation*}
firm $1$ eventually receives both workers regardless of which one it selects first. In contrast, for
\begin{equation}\label{eq:shared-firm-diverge-example}
A=
\begin{pmatrix}
10&10\\
9&0
\end{pmatrix},
\end{equation}
choosing the first worker makes firm $1$ receive the second worker in the next round, while choosing the second worker lets firm $2$ take the first worker at value $9$ in the current round. The final matchings differ.

Thus, the actual trigger condition is
\begin{equation}\label{eq:conflict-trigger}
\text{only endpoint-conflicting safe candidates can create a genuine matching choice}.
\end{equation}
Row and column ties matter precisely because two safe edges sharing a firm must form a row tie, while two safe edges sharing a worker must form a column tie. Equal values on unrelated edges require no special treatment.

\paragraph{A secondary objective after the primary guarantees are secured.}
Once a genuine or potential endpoint conflict appears, Theorem~\ref{thm:safe-round-family} already guarantees that any safe choice preserves the primary guarantees: exact EF1, $\alpha\ge\Gamma_m(t)$, and $\SW/\OPT\ge u_m(t)$. The remaining problem is therefore a secondary optimization problem, not a feasibility problem.

We use coalition stability as the secondary objective because $\Phi$ is exactly equivalent to the stability measure used in the paper:
\begin{equation}\label{eq:secondary-Phi-meaning}
\Phi(X)=\min_{z\ge0}\max_iR_i(X,z)=\frac{1}{\alpha(X)}.
\end{equation}
For a fixed personnel allocation $X$, $\Phi(X)$ is the minimum normalized stabilization pressure borne by the worst firm after wages or internal compensation are optimally redesigned. Therefore,
\begin{equation*}
\Phi(X^A)<\Phi(X^B)
\quad\Longleftrightarrow\quad
\alpha(X^A)>\alpha(X^B)
\end{equation*}
means that allocation $X^A$ is easier to stabilize than $X^B$. Welfare $\SW$ measures total value creation, while $\Phi$ measures how difficult it is to make all firms accept the personnel allocation without a strong coalition deviation. The two criteria need not agree when $m\ge3$.

Because $\Phi$ is defined for complete matchings while a safe-round state commits only part of the workers, we introduce a common \emph{benchmark completion used only for scoring}. For every positive-value worker $j$, fix a highest-value firm
\begin{equation}\label{eq:benchmark-employer}
\tau(j)\in\argmaxop_{i\in\F}a_{ij},
\end{equation}
using a fixed deterministic tie-breaking rule when necessary. At decision time $t$, let $X^t=(X_i^t)_i$ be the permanently committed partial bundles and let $U_t$ be the uncommitted positive-value workers. Define the complete comparison matching
\begin{equation}\label{eq:Yt-benchmark-completion}
Y_i^t
=
X_i^t\cup\{j\in U_t:\tau(j)=i\},
\qquad i\in\F.
\end{equation}
All-zero workers are placed by a fixed rule and do not affect the scores.

The comparison matching $Y^t$ must be distinguished from the actual algorithm state. The partial matching $X^t$ contains permanent commitments. The matching $Y^t$ is neither a new permanent allocation nor a prediction of the terminal output. It is only a common complete benchmark: committed workers stay where they are, while uncommitted workers are temporarily placed at a highest-value firm. The actual algorithm still commits only the selected safe edge. At the next decision point, a new benchmark completion $Y^{t+1}$ is constructed by the same rule.

For a candidate $e=(i,j)$ in a nontrivial conflict component $\mathcal K_t$, let $k=\mu_{Y^t}(j)=\tau(j)$. Define $Y^{t,e}$ by changing only worker $j$'s comparison employer in $Y^t$, from benchmark employer $k$ to candidate firm $i$. All other workers remain in the same comparison positions. If $k=i$, then $Y^{t,e}=Y^t$. Thus $Y^{t,e}$ isolates the one-worker change created by committing edge $e$ relative to the common market benchmark.

The exact secondary rule is
\begin{equation}\label{eq:exact-secondary-rule}
e_t^*\in
\operatorname*{arg\,min}_{e\in\mathcal K_t}
\Phi(Y^{t,e}).
\end{equation}
Different conflict components have disjoint endpoints and need not be compared with each other. Under the common benchmark completion, each conflicting safe candidate induces a complete comparison matching whose wage vector is fully re-optimized. The rule selects the candidate with the largest core factor, equivalently the smallest $\Phi$. It is therefore a one-step secondary stability optimum relative to the common benchmark completion. It is not a prediction of the terminal matching and does not claim global path optimality.

This rule can be used as the optional subroutine in Algorithm~\ref{alg:safe-round-optional}:
\begin{algorithm}[ht]
\caption{\textsc{Stability-Aware Conflict Resolver}}\label{alg:stability-resolver}
\begin{algorithmic}[1]
\Require Permanently committed bundles $X^t$, uncommitted workers $U_t$, and a nontrivial conflict component $\mathcal K_t$
\State Construct the complete scoring benchmark $Y^t$ using~\eqref{eq:Yt-benchmark-completion}
\If{exact $\Phi$ mode is used}
  \ForAll{$e=(i,j)\in\mathcal K_t$}
    \State Construct $Y^{t,e}$ by changing only worker $j$'s comparison employer in $Y^t$
    \State $Q_t(e)\gets\Phi(Y^{t,e})$
  \EndFor
\Else
  \State Compute $z^t\in\argminop_{z\ge0}\max_rR_r(Y^t,z)$
  \ForAll{$e=(i,j)\in\mathcal K_t$}
    \State Set $k\gets\mu_{Y^t}(j)$
    \State $Q_t(e)\gets\Phi(Y^t)$ if $k=i$; otherwise $Q_t(e)\gets\widehat\Phi_{j:k\to i}(z^t)$
  \EndFor
\EndIf
\State \Return $e_t^*\in\argminop_{e\in\mathcal K_t}Q_t(e)$, using the fixed tie-breaking rule if needed
\end{algorithmic}
\end{algorithm}

The exact-$\Phi$ mode solves the fixed-matching LP for each conflict candidate. The certificate mode solves for optimal wages only once at the common benchmark $Y^t$ and then uses the one-worker fixed-wage sensitivity formula from Section~\ref{sec:sensitivity}. Since $Y^{t,e}$ differs from $Y^t$ only in worker $j$'s comparison employer, $\widehat\Phi_{j:k\to i}(z^t)$ has a direct interpretation: it is the one-step worst-firm bottleneck certificate created by committing candidate edge $(i,j)$ while the benchmark-optimal wages are held fixed. It is cheaper to compute but is only an upper-bound certificate for $\Phi(Y^{t,e})$, not the exact secondary objective.

The roles of $\Phi$ and $\widehat\Phi$ are therefore distinct:
\begin{equation}\label{eq:secondary-positioning}
\begin{aligned}
&\text{no endpoint conflict: order only; no stability LP is needed;}\\
&\text{endpoint conflict: use exact }\Phi(Y^{t,e})\text{ under the common benchmark;}\\
&\qquad \widehat\Phi\text{ is a lower-cost one-step certificate;}\\
&\text{either way: no claim of global path optimality.}
\end{aligned}
\end{equation}

This also explains why we use $\Phi$ rather than another score. The safe-round structure already separates primary guarantees from a secondary objective. If a planner instead prioritizes welfare, workload balance, or another operational criterion, that score can be used among conflict candidates. We use $\Phi$ because the paper studies compatibility between fairness and coalition stability, and $\Phi=1/\alpha$ is the exact stability difficulty of a fixed matching rather than a heuristic score.

\paragraph{Relation to abstract confluence.}
In a finite-state dynamic process, uniqueness of the terminal state can equivalently be phrased by saying that all reachable branches eventually reconverge. This abstract characterization is mathematically valid but is not our algorithmic test. Lemma~\ref{lem:safe-disjoint-commute} and Theorem~\ref{thm:strict-rank-unique} are more useful here: they provide, respectively, an online conflict test and an ex-ante uniqueness condition, and they identify exactly when a secondary stability objective can have decision value.

\section{Main result: the two-dimensional EF1--core minimax frontier}\label{sec:frontier}

We now lift the fixed-matching representation and constructive lower bounds to a cross-instance performance frontier. Unlike a standard approximation ratio for a particular algorithm, we first select the \emph{most stable} EF1 matching in each instance and then take the worst case over instances. The frontier therefore measures the unavoidable stability loss caused by the fairness requirement itself, rather than the performance loss of a specific algorithm.

Define
\begin{equation}\label{eq:Hm-def}
H_m(\delta)=\inf_{\substack{|\F|=m\\\delta(A)\ge\delta}}
\max_{X:\,\mathrm{EF1}}\alpha(X).
\end{equation}
For any fixed $\beta>0$, also define the positive-fairness frontier over arbitrary market sizes by
\begin{equation}\label{eq:Cbeta}
\mathcal C_\beta(\delta)=
\inf_{m\ge2}\inf_{\delta(A)\ge\delta}
\max_{X:\,\beta\text{-EF1}}\alpha(X).
\end{equation}

\begin{table}[ht]
\centering
\caption{Main minimax results for the two-dimensional EF1--core frontier $H_m(\delta)$. Equality indicates a closed region; $g_3(\delta)=(1+\delta+2\delta^2)/[2(1+\delta)]$.}
\label{tab:frontier}
\small
\setlength{\tabcolsep}{4.5pt}
\begin{tabularx}{\textwidth}{>{\centering\arraybackslash}p{0.10\textwidth} >{\centering\arraybackslash}p{0.12\textwidth} >{\centering\arraybackslash}p{0.28\textwidth} >{\centering\arraybackslash}p{0.34\textwidth} X}
\toprule
& $\delta=0$ & $0<\delta\le1/2$ & $1/2<\delta<1$ & $\delta=1$\\
\midrule
$m=2$ & $1/2$ & \multicolumn{2}{c}{$H_2(\delta)=(1+\delta)/2$} & $1$\\
\midrule
$m=3$ & $1/3$ & $H_3(\delta)=(1+2\delta)/3$ & $g_3(\delta)\le H_3(\delta)\le(1+2\delta)/3$ & $1$\\
\midrule
$m\ge4$ & $1/m$ & \multicolumn{2}{c}{$\Gamma_m(\delta)\le H_m(\delta)\le u_m(\delta)$} & $1$\\
\midrule
$m\to\infty$ & $0$ & $\delta$ & $\delta$ & $1$\\
\bottomrule
\end{tabularx}
\end{table}

\begin{theorem}[Finite-firm frontier]\label{thm:Hm}
For every $m\ge2$ and $\delta\in[0,1]$,
\begin{equation}\label{eq:Hm-bounds}
\Gamma_m(\delta)\le H_m(\delta)\le u_m(\delta).
\end{equation}
In particular, $H_m(0)=1/m$, and for fixed $\delta$, $H_m(\delta)\to\delta$ as $m\to\infty$.
\end{theorem}
\begin{proof}
The lower bound follows directly from Theorem~\ref{thm:round-guarantee}. Every instance with $\delta(A)\ge\delta$ admits an exact-EF1 output of the round algorithm with
$\alpha\ge\Gamma_m(\delta(A))\ge\Gamma_m(\delta)$.

For the upper bound, consider $m$ firms and $m$ identical workers. Firm $1$ values each worker at $1$, while every other firm values each worker at $\delta$. When $\delta=0$, first replace $0$ by an arbitrary $\rho>0$ and then let $\rho\downarrow0$. For $\delta>0$, no EF1 matching can leave a firm empty. If one firm were empty, another firm would contain at least two workers because the number of workers equals the number of firms. After removing one worker from that bundle, the empty firm would still assign positive value to the remaining worker, violating EF1. Hence every firm receives exactly one worker, and the welfare ratio is
\begin{equation*}
\frac{1+(m-1)\delta}{m}=\delta+\frac{1-\delta}{m}=u_m(\delta).
\end{equation*}
By Proposition~\ref{prop:exact}, the core factor of any EF1 matching is at most this welfare ratio, giving $H_m(\delta)\le u_m(\delta)$. The same upper bound follows at $\delta=0$ by taking $\rho\downarrow0$.

At $\delta=0$, $\Gamma_m(0)=u_m(0)=1/m$, so $H_m(0)=1/m$. For fixed $\delta$, we have $\Gamma_m(\delta)\ge\delta$ and $u_m(\delta)=\delta+(1-\delta)/m\to\delta$, which proves the limit.
\end{proof}

\begin{corollary}[Two-firm frontier]\label{cor:H2}
For all $\delta\in[0,1]$,
\begin{equation}\label{eq:H2}
H_2(\delta)=\frac{1+\delta}{2}.
\end{equation}
\end{corollary}
\begin{proof}
For $m=2$, Theorem~\ref{thm:round-guarantee} gives an EF1 matching with welfare ratio at least $u_2(\delta)=(1+\delta)/2$. By Theorem~\ref{thm:two}, the stability rate of that matching equals its welfare ratio, so $H_2(\delta)\ge(1+\delta)/2$. The symmetric upper-bound instance in Theorem~\ref{thm:Hm} gives the reverse inequality.
\end{proof}

\begin{theorem}[Three-firm frontier]\label{thm:H3}
Let
\begin{equation}\label{eq:g3}
g_3(\delta)=\frac{1+\delta+2\delta^2}{2(1+\delta)},\qquad
u_3(\delta)=\frac{1+2\delta}{3}.
\end{equation}
For $0\le\delta\le1/2$,
$H_3(\delta)=u_3(\delta)$. For $1/2\le\delta\le1$,
\begin{equation}\label{eq:H3-bounds}
g_3(\delta)\le H_3(\delta)\le u_3(\delta).
\end{equation}
The maximum gap between the two bounds is below $0.01197$.
\end{theorem}
\begin{proof}
The upper bound again follows from the symmetric three-firm instance in Theorem~\ref{thm:Hm}. The lower bound requires a sharper analysis of the three-firm round structure; the full calculation is given in Appendix~\ref{app:H3}. The appendix shows that every complete three-worker round of Algorithm~\ref{alg:round} supports an
\begin{equation*}
f_3(\delta)=\min\{u_3(\delta),g_3(\delta)\}
\end{equation*}
-core. A final round with fewer than three workers supports at least $(1+\delta)/2\ge f_3(\delta)$, and wages and profits can be added across rounds. Hence the full algorithm supports at least $f_3(\delta)$.

When $\delta\le1/2$, the extreme-ray comparison in the appendix gives $f_3(\delta)=u_3(\delta)$, so the bounds close. When $\delta\ge1/2$, $f_3(\delta)=g_3(\delta)$. The gap is
\begin{equation*}
u_3(\delta)-g_3(\delta)=\frac{(1-\delta)(2\delta-1)}{6(1+\delta)},
\end{equation*}
which is maximized at $\delta=\sqrt3-1$ and is below $0.01197$.
\end{proof}

\begin{theorem}[Tight scale-free frontier]\label{thm:scale}
For every fixed $\beta\in(0,1]$ and $\delta\in[0,1]$,
\begin{equation}\label{eq:scale}
\mathcal C_\beta(\delta)=\delta.
\end{equation}
\end{theorem}
\begin{proof}
For the lower bound, Algorithm~\ref{alg:round} outputs exact EF1 and therefore satisfies $\beta$-EF1 for every $\beta\le1$. By the positive-edge certificate, it supports at least a $\delta(A)\ge\delta$ core. Hence $\mathcal C_\beta(\delta)\ge\delta$.

For the upper bound, use the same symmetric instance with $m$ firms and $m$ identical workers: firm $1$ values each worker at $1$, and all other firms value each worker at $\delta$. For $\delta>0$, any $\beta$-EF1 matching with $\beta>0$ must assign one worker to each firm by the same argument as in Theorem~\ref{thm:Hm}. The welfare upper bound therefore gives a stability rate of at most $\delta+(1-\delta)/m$. Letting $m\to\infty$ yields $\mathcal C_\beta(\delta)\le\delta$. For $\delta=0$, use a positive $\rho$ and then let $\rho\downarrow0$.
\end{proof}

\section[Extension I: stronger fairness---EFX guarantees at the delta-core benchmark]{Extension I: stronger fairness---EFX guarantees at the $\delta$-core benchmark}\label{sec:efx}

The scale-free EF1--core frontier is already closed at $\delta$. We therefore use the $\delta$-core as the stability benchmark that should not be weakened and ask how much stronger fairness can be guaranteed within the same financing framework. For nonnegative valuations, we use the $\EFXp$ version in Definition~\ref{def:fairness}, which requires the comparison after removal only for workers $g$ with $a_{ig}>0$. The purpose is to test robustness of the framework, not to claim that the general $\EFXp$ frontier is fully characterized.

Define
\begin{equation}\label{eq:Fefx}
\mathfrak F_{\EFXp}(\delta)=
\inf_{\delta(A)\ge\delta}
\max_{\substack{X:\,\alpha(X)\ge\delta}}
\sup\{\gamma:X\text{ is }\gamma\text{-}\EFXp\}.
\end{equation}

\subsection{Qualified-load algorithm}
For a threshold $\tau\in(0,1]$, define worker $j$'s qualified-firm set by
\begin{equation}\label{eq:eligible}
E_j(\tau)=\{i:a_{ij}\ge\tau M_j\}.
\end{equation}
The $M$-load of a firm is the sum of $M_j$ over its current workers.

\begin{algorithm}[ht]
\caption{Descending-weight minimum-qualified-load algorithm}\label{alg:qualified-load}
\begin{algorithmic}[1]
\State Order all workers with $M_j>0$ by nonincreasing $M_j$
\State Set $X_i\gets\varnothing$ and load $L_i\gets0$ for every firm
\For{each worker $j$ in this order}
  \State Choose a firm $i^*\in E_j(\tau)$ with minimum current $L_i$
  \State Set $X_{i^*}\gets X_{i^*}\cup\{j\}$ and $L_{i^*}\gets L_{i^*}+M_j$
\EndFor
\State Assign all-zero workers arbitrarily
\end{algorithmic}
\end{algorithm}

For observer firm $i$ and worker set $T$, define the qualification benchmark
\begin{equation}\label{eq:benchmark}
B_i^\tau(T)=\sum_{j\in T:\,i\in E_j(\tau)}M_j.
\end{equation}

\begin{lemma}[EFX with respect to the qualification benchmark]\label{lem:benchmark-efx}
The output of Algorithm~\ref{alg:qualified-load} satisfies, for every pair of firms $i,k$ and every $g\in X_k$ with $i\in E_g(\tau)$,
\begin{equation}\label{eq:benchmark-efx}
B_i^\tau(X_i)\ge B_i^\tau(X_k\setminus\{g\}).
\end{equation}
\end{lemma}
\begin{proof}
Fix $i,k$ and let $S_{ik}=\{j\in X_k:i\in E_j(\tau)\}$. If $S_{ik}$ is empty, the right-hand side is zero. Otherwise, let $h$ be the last worker in $S_{ik}$ processed by the algorithm. Since $i\in E_h(\tau)$, firm $i$ was also eligible when $h$ was assigned. Because $h$ was assigned to $k$, the loads immediately before assignment satisfy $L_k^-\le L_i^-$.

Before $h$ is processed, firm $k$ has already received all workers in $S_{ik}\setminus\{h\}$. Therefore,
\begin{equation*}
B_i^\tau(X_k\setminus\{h\})\le L_k^-\le L_i^-\le B_i^\tau(X_i).
\end{equation*}
The last inequality holds because workers assigned to firm $i$ are all qualified for $i$, and the final load is at least the load at the time $h$ is assigned. Since workers are processed in nonincreasing $M_j$, the last worker $h$ has the smallest $M$-weight in $S_{ik}$. Thus for any other $g\in S_{ik}$,
\begin{equation*}
B_i^\tau(X_k\setminus\{g\})\le B_i^\tau(X_k\setminus\{h\}).
\end{equation*}
Combining the inequalities proves the claim.
\end{proof}

\begin{theorem}[Polynomial-time $\delta$-$\EFXp$ benchmark]\label{thm:qualified-load}
Set $\tau=\delta(A)$. Algorithm~\ref{alg:qualified-load} runs in polynomial time and outputs a $\delta(A)$-$\EFXp$ matching that supports a $\delta(A)$-core and has welfare ratio at least $\delta(A)$.
\end{theorem}
\begin{proof}
Let $\delta_A=\delta(A)$. By definition, whenever $a_{ij}>0$ we have $a_{ij}\ge\delta_AM_j$, so $E_j(\delta_A)$ is exactly the positive-value support of worker $j$. For any worker set $T$,
\begin{equation}\label{eq:benchmark-value}
\delta_A B_i^{\delta_A}(T)\le V_i(T)\le B_i^{\delta_A}(T),
\end{equation}
because $\delta_AM_j\le a_{ij}\le M_j$ for qualified workers and firm $i$ has value zero for unqualified workers.

Fix firms $i,k$ and any $g\in X_k$ with $a_{ig}>0$. By Lemma~\ref{lem:benchmark-efx} and~\eqref{eq:benchmark-value},
\begin{align*}
V_i(X_i)
&\ge\delta_A B_i^{\delta_A}(X_i)\\
&\ge\delta_A B_i^{\delta_A}(X_k\setminus\{g\})\\
&\ge\delta_A V_i(X_k\setminus\{g\}).
\end{align*}
Thus the output is $\delta_A$-$\EFXp$. Every worker is assigned within its positive support, so Lemma~\ref{lem:positive-edge} gives a $\delta_A$-core and the same welfare guarantee. Sorting and minimum-load selection are polynomial-time operations.
\end{proof}

\subsection[QREC: raising the EFX guarantee to 1/2]{QREC: raising the EFX guarantee to $1/2$}
For a partial matching $X$, define the firm envy graph $G_X$: there is a directed edge $i\to k$ if $V_i(X_i)<V_i(X_k)$. For worker $g$, let the positive support be $S(g)=\{i:a_{ig}>0\}$. QREC maintains an acyclic envy graph, so the induced graph $G_X[S(g)]$ has a source. Here a source has indegree zero: no firm in the support envies that source firm's bundle.

\begin{algorithm}[ht]
\caption{QREC: support-source assignment, reset, and envy-cycle elimination}\label{alg:qrec}
\begin{algorithmic}[1]
\State Set $X_i\gets\varnothing$ for all firms and $U\gets\{j:M_j>0\}$
\While{$U\neq\varnothing$}
  \State Take a worker $g$ from $U$ and choose a source firm $k$ in $G_X[S(g)]$
  \State Remove $g$ from $U$ and tentatively set $X_k\gets X_k\cup\{g\}$
  \If{the new matching violates $1/2$-$\EFXp$}
    \State Choose a newly violating firm $i$ with respect to $X_k$
    \State Remove $g$ from $X_k$, return all workers in the old $X_i$ to $U$
    \State Reset $X_i\gets\{g\}$
  \EndIf
  \While{the envy graph contains a directed cycle $i_1\to i_2\to\cdots\to i_q\to i_1$}
    \State Rotate complete bundles along the cycle so that each $i_s$ receives the old bundle of $i_{s+1}$ that it envies
    \State Return to $U$ any worker whose new current employer values that worker at zero
  \EndWhile
\EndWhile
\State Assign all-zero workers arbitrarily
\end{algorithmic}
\end{algorithm}

\begin{lemma}[Support-source reset]\label{lem:qrec-reset}
Suppose QREC is at an acyclic state satisfying EF1 and $1/2$-$\EFXp$, and worker $g$ is tentatively added to a source firm $k$. If this creates a new $1/2$-$\EFXp$ violation by firm $i$, then $a_{ig}>V_i(X_i)$. Resetting firm $i$ to the singleton $\{g\}$ and returning its old bundle does not violate EF1 or $1/2$-$\EFXp$.
\end{lemma}
\begin{proof}
All fairness constraints hold before adding $g$. After the addition, all target bundles except $X_k\cup\{g\}$ remain unchanged, and firm $k$'s own value only increases. Thus a new violation can only come from some $i\neq k$ comparing itself with the new target bundle. If $a_{ig}=0$, firm $i$'s value for the target bundle does not change, so no new violation is possible. Hence $a_{ig}>0$ and $i\in S(g)$.

Because $k$ is a source in $G_X[S(g)]$, there is no envy edge from $i$ to $k$, and therefore $V_i(X_k)\le V_i(X_i)$ before $g$ is added. Removing the new worker $g$ restores the old bundle $X_k$, so deleting $g$ cannot itself be the cause of a $1/2$-$\EFXp$ violation. Therefore some old worker $h\in X_k$ must satisfy
\begin{align*}
V_i(X_i)
&<\frac12\bigl(V_i(X_k)-a_{ih}+a_{ig}\bigr)\\
&\le\frac12\bigl(V_i(X_i)+a_{ig}\bigr),
\end{align*}
which implies $a_{ig}>V_i(X_i)$.

After the reset, firm $i$'s own value strictly increases. Firm $k$ returns to its pre-addition bundle, so all fairness comparisons with target $k$ are restored. For any other firm comparing itself with the singleton target $\{g\}$, deleting $g$ leaves an empty bundle, so both EF1 and $1/2$-$\EFXp$ hold automatically. Since firm $i$'s own value increases, all of its previous fairness constraints as an observer become easier. Thus the reset preserves both fairness properties.
\end{proof}

\begin{theorem}[QREC joint guarantee]\label{thm:qrec}
Algorithm~\ref{alg:qrec} terminates after finitely many steps and outputs a complete matching that assigns every nonzero worker along a positive edge, satisfies exact EF1 and $1/2$-$\EFXp$, supports a $\delta(A)$-core, and has welfare ratio at least $\delta(A)$. If all valuations are integers, the algorithm has a pseudopolynomial bound in $\OPT$ on the number of structural updates.
\end{theorem}
\begin{proof}
We maintain four invariants at the start of each outer iteration: the current partial matching satisfies exact EF1 and $1/2$-$\EFXp$; every assigned nonzero worker lies on a positive edge; and the envy graph is a DAG. The empty matching clearly satisfies these properties.

Consider worker $g$. Since $G_X[S(g)]$ is a DAG, it has a source firm $k$. Tentatively add $g$ to $k$. For any observer firm $i$, if $a_{ig}=0$, its value for the target bundle does not change. If $a_{ig}>0$, then $i\in S(g)$ and the source property gives $V_i(X_i)\ge V_i(X_k)$ before the addition; after deleting the new worker $g$, no envy remains. Thus direct addition always preserves EF1. Since $k\in S(g)$, the assignment is along a positive edge.

If no $1/2$-$\EFXp$ violation appears, both fairness properties remain valid. If a violation appears, Lemma~\ref{lem:qrec-reset} allows us to reset a violating firm $i$ to $\{g\}$ and return its old bundle. Both fairness properties remain valid, and $a_{ig}>0$, so positive-edge eligibility is preserved.

Next eliminate envy cycles. Suppose the envy graph contains a directed cycle $i_1\to i_2\to\cdots\to i_q\to i_1$. Rotate bundles so that each firm $i_s$ receives the old bundle of $i_{s+1}$ that it strictly envies. Every firm on the cycle strictly increases its own value; firms off the cycle keep their own bundles; and the multiset of target bundles is only permuted. Therefore no EF1 or $1/2$-$\EFXp$ comparison becomes worse. If a worker is moved to an employer that values that worker at zero, return the worker to $U$. This does not reduce the current employer's own value and only shrinks a target bundle, so fairness becomes weakly easier; all remaining assigned workers still lie on positive edges. Repeating the cycle elimination restores a DAG and all four invariants.

For termination, compare social welfare at the start and end of each outer iteration. If worker $g$ is accepted without a reset, welfare increases by the positive amount $a_{kg}>0$. If a reset occurs, firm $k$ returns to its pre-iteration bundle, while firm $i$ changes from value $V_i(X_i)$ to $a_{ig}>V_i(X_i)$, so welfare strictly increases relative to the start of the iteration. Every envy-cycle rotation also strictly increases the own value of every firm on the cycle, while returning zero-valued workers does not change welfare. Hence welfare strictly increases after every outer iteration. The number of partial matching states is finite, so the algorithm cannot run forever. With integer valuations, every strict welfare increase is at least $1$ and $\SW\le\OPT$, giving a pseudopolynomial bound of order $\OPT$ on structural updates.

At termination, every worker with $M_j>0$ is assigned on a positive edge. The invariants give exact EF1 and $1/2$-$\EFXp$. Lemma~\ref{lem:positive-edge} then gives a $\delta(A)$-core and the same welfare guarantee. Adding all-zero workers does not affect values or fairness.
\end{proof}

\begin{theorem}[EFX lower bound at the optimal scale-free core benchmark]\label{thm:efx-frontier}
For every $\delta\in(0,1]$,
\begin{equation}\label{eq:efx-frontier}
\mathfrak F_{\EFXp}(\delta)\ge\max\{\delta,1/2\}.
\end{equation}
Moreover, $\mathfrak F_{\EFXp}(1)=1$.
\end{theorem}
\begin{proof}
For an instance with $\delta(A)\ge\delta$, Algorithm~\ref{alg:qualified-load} produces at least a $\delta$-$\EFXp$ matching supporting at least a $\delta$-core. Algorithm~\ref{alg:qrec} produces a $1/2$-$\EFXp$ matching that supports a $\delta(A)$-core and hence also a $\delta$-core. Selecting, for each instance, the output with the larger fairness factor gives the uniform lower bound $\max\{\delta,1/2\}$. At $\delta=1$, set $\tau=1$ in the qualified-load algorithm. The two inequalities in~\eqref{eq:benchmark-value} become equalities, giving exact $\EFXp$ together with the exact core.
\end{proof}

For general $0<\delta<1$, the exact value of $\mathfrak F_{\EFXp}(\delta)$ remains open. This should be distinguished from the closed scale-free EF1--core frontier.

\section[Extension II: capacities and Top-r financing]{Extension II: capacities and Top-$r$ financing}\label{sec:capacity}

Capacity constraints provide a second robustness test. They change not only the feasible matchings but also the deviation set: a firm can recruit at most $r_i$ workers in a coalition deviation. The stabilization representation must therefore replace the sum of all positive deviation terms by the sum of the largest at most $r_i$ terms. Candidate workers still come from the entire set $\W$ and may include current workers; we do not assume that current workers automatically have zero positive terms in the capacity model. In the unconstrained model, after a firm exits a round it can face at most $m-1$ competitive directions. With capacity $r_i$, the effective propagation multiplicity is truncated at $q_i=\min\{r_i,m-1\}$. The capacity results therefore test the same main mechanism: how branching changes the firm-level stabilization bottleneck.

Firm $i$ now has capacity $r_i\in\mathbb Z_+$. A capacity-feasible matching $\mu$ satisfies $|\mu(i)|\le r_i$; unmatched workers are treated as assigned to a zero-value dummy node. Let $\OPTr$ denote optimal welfare under capacities. For a coalition of firms $I$ and worker set $T$, the capacity-constrained deviation value is the maximum total value obtained by assigning workers in $T$ to firms in $I$, with each firm receiving at most $r_i$ workers.

\begin{definition}[Capacity EF1]
A capacity-feasible matching $\mu$ satisfies capacity EF1 if, for every pair of real firms $i,k$, there exists $g\in\mu(k)$ such that
\begin{equation}\label{eq:cap-ef1}
V_i(\mu(i))\ge
\max_{\substack{T\subseteq\mu(k)\setminus\{g\}\\|T|\le r_i}}V_i(T).
\end{equation}
\end{definition}

For a nonnegative vector $u=(u_j)$, let $\Top_r(u)$ be the sum of its largest at most $r$ coordinates. For a fixed capacity matching $\mu$, set $d_j=a_{\mu(j)j}$, with $d_j=0$ for unmatched workers, and let $B_i=V_i(\mu(i))$. Define
\begin{equation}\label{eq:cap-D}
D_i^{(r)}(\mu,z)=\sum_{j\in\mu(i)}z_j+
\Top_{r_i}\bigl((\pos{a_{ij}-z_j})_{j\in\W}\bigr).
\end{equation}

\begin{theorem}[Capacity financing--stability identity]\label{thm:cap-identity}
Let
\begin{equation}\label{eq:cap-Phi}
\Phi_r(\mu)=\min_{z\ge0}\max_i\frac{D_i^{(r)}(\mu,z)}{B_i}.
\end{equation}
Then the optimal core factor of the fixed capacity matching is
\begin{equation}\label{eq:cap-id}
\alpha_r(\mu)=\frac{1}{\Phi_r(\mu)}.
\end{equation}
\end{theorem}
\begin{proof}
Fix firm $i$ and actual wages $y$. Because of capacity $r_i$, a singleton deviation by firm $i$ can select at most $r_i$ workers from the full set $\W$, including current workers. The most profitable set therefore consists of the largest at most $r_i$ values among $\pos{\alpha a_{ij}-y_j}$. The single-firm core condition is equivalent to
\begin{equation*}
x_i\ge\Top_{r_i}\bigl((\pos{\alpha a_{ij}-y_j})_j\bigr).
\end{equation*}
These conditions are still sufficient for all multi-firm coalition constraints. Any capacity-feasible coalition deviation partitions its workers into disjoint sets $T_i$ with $|T_i|\le r_i$; summing the firm inequalities covers the coalition value.

Write $y_j=\alpha z_j$ and use local financing
$x_i=B_i-\alpha\sum_{j\in\mu(i)}z_j$. We obtain
\begin{equation*}
B_i\ge\alpha\left[\sum_{j\in\mu(i)}z_j+
\Top_{r_i}\bigl((\pos{a_{ij}-z_j})_j\bigr)\right]
=\alpha D_i^{(r)}(\mu,z).
\end{equation*}
The remainder is identical to Theorem~\ref{thm:identity}: optimizing over wages gives the minimum worst stabilization ratio $\Phi_r$, and the best stability rate is its reciprocal.
\end{proof}

Define capacity rematching pressure by
\begin{equation}\label{eq:cap-pressure}
P_i^{(r_i)}(\mu)=
\Top_{r_i}\bigl((\pos{a_{ij}-d_j})_{j\in\W}\bigr),\qquad
\kappa_r(\mu)=\max_i\frac{P_i^{(r_i)}(\mu)}{B_i}.
\end{equation}

\begin{lemma}[Capacity proportional-wage bridge]\label{lem:cap-bridge}
Every capacity matching $\mu$ supports a $1/(1+\kappa_r(\mu))$-core. Moreover,
\begin{equation}\label{eq:cap-welfare-bridge}
\frac{\SW(\mu)}{\OPTr}\ge\frac{1}{1+\kappa_r(\mu)}.
\end{equation}
\end{lemma}
\begin{proof}
Let $\alpha=1/(1+\kappa_r)$, set $y_j=\alpha d_j$, and set $x_i=(1-\alpha)B_i$. Firm $i$'s largest capacity-constrained gap is
\begin{equation*}
\Top_{r_i}\bigl((\pos{\alpha a_{ij}-\alpha d_j})_j\bigr)
=\alpha P_i^{(r_i)}(\mu)
\le\alpha\kappa_rB_i=(1-\alpha)B_i=x_i.
\end{equation*}
The capacity single-firm conditions therefore give an $\alpha$-core.

For welfare, let $\mu^*$ be a capacity-optimal matching. Group positive welfare gains relative to $\mu$ by their destination firms in $\mu^*$. Firm $i$ receives at most $r_i$ workers in $\mu^*$, so its total positive gain is at most $P_i^{(r_i)}(\mu)$. Hence
\begin{align*}
\OPTr-\SW(\mu)
&\le\sum_iP_i^{(r_i)}(\mu)\\
&\le\kappa_r\sum_iB_i
=\kappa_r\SW(\mu).
\end{align*}
Rearranging gives~\eqref{eq:cap-welfare-bridge}.
\end{proof}

\begin{algorithm}[ht]
\caption{Capacity-constrained maximum-edge round algorithm}\label{alg:capacity-round}
\begin{algorithmic}[1]
\State Set $\mu(i)\gets\varnothing$, remaining capacity $c_i\gets r_i$, and $U\gets\W$
\While{there exist $i,j$ with $c_i>0$, $j\in U$, and $a_{ij}>0$}
  \State Start a new round and set $S\gets\{i:c_i>0\}$
  \While{there exist $i\in S$ and $j\in U$ with $a_{ij}>0$}
    \State Choose a current maximum positive edge $(i^*,j^*)$
    \State Assign $j^*$ to $i^*$, set $c_{i^*}\gets c_{i^*}-1$, and remove $j^*$ from $U$
    \State Remove $i^*$ from the current round: $S\gets S\setminus\{i^*\}$
  \EndWhile
\EndWhile
\State Leave remaining workers unmatched or assign them to the dummy node
\end{algorithmic}
\end{algorithm}

\begin{lemma}[Common removal witness with capacities]\label{lem:cap-witness}
Let $p_k^1$ be the first positive-value worker received by firm $k$. For every pair of firms $i,k$ and every set $T\subseteq\mu(k)\setminus\{p_k^1\}$ with $|T|\le r_i$,
$V_i(T)\le V_i(\mu(i))$. Hence Algorithm~\ref{alg:capacity-round} outputs a capacity-EF1 matching.
\end{lemma}
\begin{proof}
Order the workers in $T$ by the rounds in which firm $k$ receives them, say $p_k^{s_1},\ldots,p_k^{s_q}$, where $q\le r_i$. Since the first worker is removed, $s_t\ge t+1$. If firm $i$ receives a positive-value worker $p_i^t$ in round $t$, then $p_k^{s_t}$ is still unassigned when $p_i^t$ is selected, and firm $i$ still has capacity and is active in that round. The maximum-edge rule gives
$a_{i,p_i^t}\ge a_{i,p_k^{s_t}}$.

If firm $i$ has no $t$th positive-value worker, then either its capacity is already full---which cannot occur for $t\le q\le r_i$---or from that round onward it has no positive edge to any unassigned worker. In the latter case $a_{i,p_k^{s_t}}=0$. Therefore,
\begin{equation*}
V_i(T)=\sum_{t=1}^q a_{i,p_k^{s_t}}
\le\sum_{t=1}^q a_{i,p_i^t}
\le V_i(\mu(i)).
\end{equation*}
Taking the maximum over all feasible $T$ gives capacity EF1.
\end{proof}

Let $q_i=\min\{r_i,m-1\}$ and $q_r=\max_iq_i$.

\begin{lemma}[Top-$r$ pressure bound under capacity rounds]\label{lem:cap-pressure}
The output of Algorithm~\ref{alg:capacity-round} satisfies
\begin{equation}\label{eq:cap-pressure-bound}
P_i^{(r_i)}(\mu)\le q_iB_i\qquad\forall i.
\end{equation}
If every matched worker with $M_j>0$ is assigned along a positive edge and $\delta(A)\ge t$, then
\begin{equation}\label{eq:cap-pressure-delta}
P_i^{(r_i)}(\mu)\le q_i(1-t)B_i.
\end{equation}
\end{lemma}
\begin{proof}
Let $b_1,\ldots,b_\ell$ be the values received by firm $i$ in the rounds in which it obtains positive-value workers, so $B_i=\sum_{s=1}^\ell b_s$. If another firm selects worker $j$ before $i$ exits a round, then $i$ is still active and the selected edge has value at least $a_{ij}$, so $\pos{a_{ij}-d_j}=0$. Positive gains can therefore come only from workers selected after $i$ exits a round, or from tail workers selected after firm $i$ has filled its capacity.

In the round in which firm $i$ receives its $s$th worker at value $b_s$, at most $m-1$ other firms choose one worker each after $i$ exits. These workers were still unassigned when $i$ selected its own worker, so their values to $i$ are at most $b_s$. If firm $i$ has remaining capacity but never again receives a positive-value worker, then at that point it has no positive edge to any remaining worker, so no later positive gain arises. If firm $i$ fills its capacity at the $r_i$th selection, then all tail workers have value to $i$ at most the last value $b_{r_i}$.

If $r_i\le m-1$, Top-$r_i$ selects at most $r_i$ terms and each term is at most $b_1\le B_i$. Hence
$P_i^{(r_i)}\le r_iB_i=q_iB_i$. If $r_i>m-1$, replicate each $b_s$ exactly $m-1$ times. Every positive gain that could enter Top-$r_i$ is coordinatewise dominated by this multiset. Tail terms after capacity is filled are also at most $b_{r_i}$, while the replicated multiset already contains at least $r_i$ terms no smaller than $b_{r_i}$. Thus the top $r_i$ sum cannot increase, and
\begin{equation*}
P_i^{(r_i)}\le(m-1)\sum_sb_s=(m-1)B_i=q_iB_i.
\end{equation*}
This proves~\eqref{eq:cap-pressure-bound}.

If positive-edge quality is at least $t$, every matched worker satisfies $d_j\ge tM_j\ge t a_{ij}$ and hence
\begin{equation*}
\pos{a_{ij}-d_j}\le(1-t)a_{ij}.
\end{equation*}
The previous domination argument then gains a common factor $1-t$, proving~\eqref{eq:cap-pressure-delta}.
\end{proof}

\begin{theorem}[Joint capacity EF1--core guarantee]\label{thm:capacity}
Let $q_r=\max_i\min\{r_i,m-1\}$. The capacity maximum-edge round algorithm outputs a capacity-EF1 matching and satisfies
\begin{equation}\label{eq:capacity-basic}
\alpha_r(\mu)\ge\frac{1}{1+q_r},\qquad
\frac{\SW(\mu)}{\OPTr}\ge\frac{1}{1+q_r}.
\end{equation}
If every worker with $M_j>0$ is assigned by a real firm along a positive edge, let $t=\delta(A)$. Then
\begin{equation}\label{eq:capacity-delta}
\min\left\{\alpha_r(\mu),\frac{\SW(\mu)}{\OPTr}\right\}
\ge\Gamma_{m,r}(t):=
\max\left\{t,\frac{1}{1+q_r(1-t)}\right\}.
\end{equation}
\end{theorem}
\begin{proof}
Capacity EF1 follows from Lemma~\ref{lem:cap-witness}. In general, Lemma~\ref{lem:cap-pressure} gives $\kappa_r(\mu)\le q_r$, and the capacity proportional-wage bridge gives~\eqref{eq:capacity-basic}.

Under positive-edge eligibility,~\eqref{eq:cap-pressure-delta} gives $\kappa_r\le q_r(1-t)$, so the same bridge yields both core factor and welfare ratio at least $1/[1+q_r(1-t)]$. The worker-wise wage certificate $y_j=tM_j$ also remains valid for capacity coalitions because capacities only reduce the value that a deviating coalition can realize; hence the stability rate is at least $t$. Likewise, if every positive-value worker is positively matched, then $\SW(\mu)\ge t\sum_jM_j\ge t\OPTr$. Taking the better of the two guarantees gives~\eqref{eq:capacity-delta}.
\end{proof}

\begin{table}[ht]
\centering
\caption{Robustness extensions beyond the main minimax frontier.}
\label{tab:extensions}
\small
\begin{tabularx}{\textwidth}{>{\raggedright\arraybackslash}p{0.20\textwidth} >{\raggedright\arraybackslash}p{0.48\textwidth} X}
\toprule
\textbf{Extension} & \textbf{Main result} & \textbf{Method and tractability}\\
\midrule
Stronger fairness $\EFXp$ & While keeping at least a $\delta$-core and welfare ratio at least $\delta$, one can guarantee $\max\{\delta,1/2\}$-$\EFXp$. At $\delta=1$, exact $\EFXp$ and the exact core are achieved. & The qualified-load algorithm gives $\delta$-$\EFXp$; QREC gives exact EF1 and $1/2$-$\EFXp$ and terminates finitely.\\
\midrule
Capacities $r_i$ & A fixed capacity matching satisfies $\alpha_r=1/\Phi_r$. The round algorithm gives capacity EF1 and baseline core/welfare guarantee $1/(1+q_r)$; with positive-edge eligibility, the guarantee improves to $\Gamma_{m,r}(\delta)$. & The deviation benchmark is the Top-$r_i$ sum of positive gains, with $q_r=\max_i\min\{r_i,m-1\}$. Fixed-matching stability remains exactly computable by linearizing Top-$r$.\\
\bottomrule
\end{tabularx}
\end{table}

\section{Operational interpretation, scope, and managerial implications}\label{sec:discussion}

\subsection[Using Ri to diagnose the stability effect of a fairness repair]{Using $R_i$ to diagnose the stability effect of a fairness repair}
The firm-bottleneck view separates the fairness--stability conflict into three levels. The first is the worker-level direct mismatch $M_j-d_j$, which measures the gap between a high market wage requirement and the worker's output after the worker is moved away from a highest-value edge. The second is the firm-level stabilization ratio
\begin{equation*}
R_i(X,z)=D_i(X,z)/B_i,
\end{equation*}
which places current wage requirements, deviation gains, and the firm's production base on the same scale. A matching transfer initially changes only the source and receiving firms' ratios. Under the old wages, the largest bottleneck may rise, fall, or remain unchanged; if the target stability rate is already met, no wage adjustment is required. The third level is the propagation created by optional wage rebalancing. Lowering a worker's wage can reduce the current employer's ratio but increase the ratios of several competitors, especially those with small output bases $B_r$. Thus the relevant fixed-wage stability statistic is $\rho(X,z)=\max_iR_i(X,z)$, while the re-optimized stability rate depends on the full wage-vector problem $\Phi(X)=\min_z\rho(X,z)$.

This distinction implies that a fairness repair with minimum welfare loss need not minimize stability loss. Welfare records the change in $\sum_jd_j$, whereas the core factor depends on how the full vector of $R_i$ changes, which firms become bottlenecks, and how well wages can rebalance those bottlenecks. The mutual-top condition handles the primary constraints first by restricting choices to a safe set with global EF1, core, and welfare guarantees. Safe candidates with disjoint endpoints only differ in order and require no extra computation; strict row and column rankings further give an ex-ante unique output. Only candidates sharing a firm or worker create genuine or potential personnel differences. In those cases, $\Phi$ has a clear operational role as a secondary objective: it measures how difficult it remains to stabilize a candidate personnel plan after compensation is optimally redesigned. Under the common benchmark completion, the optional resolver uses exact $\Phi$ for the one-step secondary comparison, while $\widehat\Phi$ provides a cheaper fixed-wage certificate. The worst-case floor remains $\Gamma_m(\delta)$ for stability and $u_m(\delta)$ for welfare.

\subsection{Settings in which firm-side EF1 is natural}
The fairness notion is most natural when several units within the same organization or platform compete for a shared talent pool. In an internal labor market, firms can be interpreted as business units and wages as actual compensation plus internal transfer prices. In a platform organization, firms can represent project teams and $a_{ij}$ the marginal output of worker $j$ in team $i$. Similar interpretations apply to staff rotation within hospitals or large institutions. In these environments, a central planner may care both about excessive concentration of high-value workers and about the recruitment or compensation pressure created after rebalancing. Firm-side EF1 provides a discrete and implementable minimum standard for bundle balance, while $\alpha(X)$ quantifies the coalition-stability cost of reaching that standard.

This interpretation also clarifies where the model is not intended to apply directly. If the main policy objective is worker-side opportunity fairness, group quotas, or justified envy, additional worker-side constraints are needed; firm-side EF1 does not replace them. Likewise, wage floors or ceilings, non-transferable worker preferences, or strong complementarities in team value would require modifications of the financing identity.

\subsection{Theoretical limits and future directions}
The model relies on four structural assumptions: nonnegative additive firm values, transferable payments, static complete information, and a firm-bundle notion of fairness. Additivity makes single-firm deviation gaps separable by worker, while transferable payments make it possible to express stability as a financing problem. These two features support~\eqref{eq:intro-id}. Extending the results to complementarities or general gross-substitutes valuations is therefore not a mechanical replacement; a different aggregation of deviation gaps would be needed.

The minimax frontier uses positive-edge quality $\delta$ as an instance parameter. In data-driven settings, one could instead seek instance-dependent frontiers based on valuation distributions, sparsity, or network structure. The mutual-top safe set also reveals an algorithmic limitation that is not obvious from the worst-case guarantees: many safe choices differ only by order, and genuine decision freedom appears only in endpoint conflicts. Strict row and column rankings give an ex-ante uniqueness condition, while row or column ties are only warnings that a conflict may arise. Two open algorithmic questions are particularly natural: can one find weaker ex-ante conditions that still guarantee a unique output, and can stability-aware one-step choices be shown to improve the terminal stability rate systematically for important classes of conflicts, perhaps strengthening the finite-$m$ frontier lower bound?

\section{Conclusion}\label{sec:conclusion}

We study a stability optimization problem created endogenously by a fairness constraint. In a many-to-one market with transferable payments, a welfare-maximizing matching can support the exact core, but firm-side EF1 may require the planner to change the personnel allocation. Rather than measuring only a welfare price of fairness, we directly optimize the coalition-core factor that a fair matching can support and express it as a tractable firm-bottleneck minimax problem.

The main technical tool is the financing--stability identity $\alpha(X)=1/\Phi(X)$. It decomposes the stability of a fixed personnel allocation into firm ratios $R_i=D_i/B_i$ and allows wages or internal compensation to be re-optimized to minimize the largest pressure. A fairness adjustment can therefore be analyzed in two stages: the direct matching effect on the source and receiving firms under old wages, followed, when needed, by the cross-firm propagation caused by compensation re-optimization. This separation gives closed-form local sensitivity, a fixed-wage feasibility test, and an exact LP that can be used as a stability diagnostic for a given allocation.

At the global level, mutual-top safe edges define a robust choice domain that preserves exact EF1, a $\Gamma_m(\delta)$ core guarantee, and a $u_m(\delta)$ welfare guarantee. Safe candidates with disjoint endpoints commute, and strict local row and column rankings imply an ex-ante unique positive-value matching. Hence extra stability optimization is not needed at every step. It is relevant only when safe candidates share a firm or a worker, in which case $\Phi$ can be used as a secondary criterion under a common benchmark completion.

The main result is the EF1--core minimax frontier. For general finite $m$,
\[
\Gamma_m(\delta)\le H_m(\delta)\le u_m(\delta).
\]
The frontier equals $1/m$ at $\delta=0$, is fully characterized for two firms, and is exact for three firms when $\delta\le1/2$, with only a small remaining gap elsewhere. When the number of firms is unbounded, the tight scale-free stability rate is exactly $\delta$. Positive-edge quality is therefore both an algorithmic parameter and a structural determinant of fairness--stability compatibility across market sizes.

From an operations-research perspective, the results suggest a layered decision rule. First, use provably safe choices to satisfy fairness and worst-case performance requirements. Second, invoke stability optimization only when a genuine allocation conflict appears. Finally, once the personnel allocation is fixed, redesign compensation to minimize the worst firm bottleneck. For internal labor markets, project-team formation, and other centrally managed talent pools, this links ``who is assigned where'' with ``how the allocation can be stabilized through transfers'' in one tractable framework. The $\EFXp$ and capacity extensions show that the same financing logic remains useful when fairness is strengthened or deviation size is restricted. Natural next steps include non-additive or gross-substitutes values, wage rigidity, data-driven instance parameters, and stronger terminal guarantees for conflict-aware local stability choices.

\appendix
\section{Full technical proof for the three-firm frontier}\label{app:H3}

This appendix gives the complete three-firm round analysis needed for the lower bound in Theorem~\ref{thm:H3}. Consider a complete round and relabel firms and workers by the selection order so that the three matches are $(1,1),(2,2),(3,3)$. Let $d_t=a_{tt}$. The maximum-edge rule gives $d_1\ge d_2\ge d_3>0$. For $p>q$, firm $p$ is still active when worker $q$ is selected, so $a_{pq}\le d_q$. For $p<q$, firm $p$ has already exited, and $\delta(A)\ge\delta$ gives
\begin{equation}\label{eq:app-pq}
a_{pq}\le\min\{d_p,d_q/\delta\}.
\end{equation}

To control the fixed-matching core exactly, use the fractional-knapsack interpretation of the dual LP. Fix nonnegative firm prices $\pi=(\pi_1,\pi_2,\pi_3)$ and define the baseline
\begin{equation*}
D=d_1\pi_1+d_2\pi_2+d_3\pi_3.
\end{equation*}
Let
\begin{equation}\label{eq:ABC}
A=\min\{d_1,d_2/\delta\}-d_2,\quad
B=\min\{d_1,d_3/\delta\}-d_3,\quad
C=\min\{d_2,d_3/\delta\}-d_3.
\end{equation}
Worker $1$ creates no additional blocking pressure relative to the current employer. Worker $2$ creates at most $A\min\{\pi_1,\pi_2\}$. For worker $3$, allocate available capacity first between firm $1$ and the current employer and then use any remaining capacity for firm $2$. The total additional pressure is at most
\begin{equation}\label{eq:Gpi}
G(\pi)=A\min\{\pi_1,\pi_2\}
+B\min\{\pi_1,\pi_3\}
+C\min\{\pi_2,\pos{\pi_3-\pi_1}\}.
\end{equation}
Thus total dual blocking value is at most $D+G(\pi)$, and the supported stability rate is at least
$1/[1+\max_\pi G(\pi)/D]$.

Normalize $d_1=1$ and write $r=d_2/d_1$, $s=d_3/d_1$. On each polyhedral cone induced by the hyperplanes $\pi_1=\pi_2$, $\pi_1=\pi_3$, and $\pi_3=\pi_1+\pi_2$, the ratio $G/D$ is a ratio of two linear functions. If, within one cone, $\pi=\sum_h\lambda_hv^h$ is a nonnegative combination of extreme rays, then
\begin{equation*}
\frac{G(\pi)}{D(\pi)}
=\sum_h\frac{\lambda_hD(v^h)}{D(\pi)}\frac{G(v^h)}{D(v^h)}.
\end{equation*}
Hence the maximum is attained on an extreme ray. To make this reduction explicit, the partition is generated only by the three internal hyperplanes above and the three coordinate planes. Every extreme ray of a three-dimensional cone is the intersection of two boundary planes. Checking these intersections and retaining nonnegative feasible directions gives, in addition to the three coordinate axes,
\begin{equation}\label{eq:rays}
(1,1,0),\quad(1,0,1),\quad(0,1,1),\quad(1,1,1),\quad(1,1,2).
\end{equation}
The first three non-axis rays give, respectively, $A/(1+r)$, $B/(1+s)$, and $C/(r+s)$. For the first term, if $r\le\delta$, then $A=r(1/\delta-1)$ and the ratio increases with $r$; if $r\ge\delta$, then $A=1-r$ and the ratio decreases with $r$. Thus the maximum occurs at $r=\delta$ and equals
\begin{equation}\label{eq:k1}
k_1=\frac{1-\delta}{1+\delta}.
\end{equation}
The same argument shows that the other two terms are also at most $k_1$.

For ray $(1,1,1)$, the value is $(A+B)/(1+r+s)$. Define
$h(t)=\min\{1,t/\delta\}-t$. Checking separately the regions $t\le\delta$ and $t\ge\delta$ shows that the ratio is maximized at $r=s=\delta$, with value
\begin{equation}\label{eq:ku}
k_u=\frac{2(1-\delta)}{1+2\delta}.
\end{equation}
For ray $(1,1,2)$, the value is $(A+B+C)/(1+r+2s)$. The numerator changes linear form only on the lines $r=\delta$, $s=\delta$, and $s=\delta r$, so the maximum of the linear-fractional function on each region occurs at a vertex. Excluding zero-value vertices, it is enough to check
\begin{center}
\begin{tabular}{cc}
\toprule
$(r,s)$ & $(A+B+C)/(1+r+2s)$\\
\midrule
$(\delta,0)$ or $(1,\delta)$ & $(1-\delta)/(1+\delta)$\\
$(\delta,\delta)$ & $2(1-\delta)/(1+3\delta)$\\
$(\delta,\delta^2)$ & $(1-\delta)(1+2\delta)/(1+\delta+2\delta^2)$\\
\bottomrule
\end{tabular}
\end{center}
The first two values are no larger than $k_u$. Let the last value be $k_t$. Then
\begin{equation}\label{eq:ku-kt}
k_u-k_t=
\frac{(\delta-1)(2\delta-1)}{(1+2\delta)(1+\delta+2\delta^2)}.
\end{equation}
Therefore the largest additional pressure is $k_u$ when $\delta\le1/2$ and $k_t$ when $\delta\ge1/2$. The corresponding core factors are
\begin{equation*}
\frac{1}{1+k_u}=\frac{1+2\delta}{3}=u_3(\delta),\qquad
\frac{1}{1+k_t}=\frac{1+\delta+2\delta^2}{2(1+\delta)}=g_3(\delta).
\end{equation*}
This proves that every complete three-worker round supports an
$f_3(\delta)=\min\{u_3(\delta),g_3(\delta)\}$-core.

If the final round contains one worker, that worker is assigned to a highest-value firm, so the round supports the exact core. If the final round contains two workers, relabel them by selection order as $(1,1),(2,2)$. Let $d_1=a_{11}$, $d_2=a_{22}$, and $M_2=\max_i a_{i2}$. The first edge is a global maximum edge, so $M_1=d_1$. When the second edge is selected, firms $2$ and $3$ are still active, so the third firm's value for worker $2$ is at most $d_2$. Define
\begin{equation*}
\rho=\frac{d_1+d_2}{d_1+M_2}.
\end{equation*}
Set $y_1=\rho d_1$, $y_2=d_2$, $x_1=(1-\rho)d_1$, and all other profits to zero. Worker $1$'s wage covers every firm's $\rho$-gap. Worker $2$ can create a gap only for firm $1$, and
\begin{equation*}
\rho M_2-d_2=(1-\rho)d_1=x_1.
\end{equation*}
Thus the round supports a $\rho$-core. Since $d_2\ge\delta M_2$ and $d_1\ge M_2$, direct simplification gives
$\rho\ge(1+\delta)/2\ge f_3(\delta)$.

Finally, add wages and firm profits across rounds. Every round supports the same core factor $f_3(\delta)$, and a firm's single-firm gaps are additive across workers from different rounds. Hence the combined matching also supports an $f_3(\delta)$-core. Together with the common removal witness, this completes the three-firm lower-bound proof.

\section{LP linearization for a fixed capacity matching}\label{app:cap-lp}
The Top-$r_i$ term in the capacity financing identity can be linearized. For any nonnegative vector $u$,
\begin{equation}\label{eq:topr-linearize}
\Top_{r_i}(u)=\min_{t_i\ge0}\left\{r_it_i+\sum_j\pos{u_j-t_i}\right\}.
\end{equation}
Substituting $u_j=\pos{a_{ij}-z_j}$ and introducing
$s_{ij}\ge a_{ij}-z_j-t_i$ and $s_{ij}\ge0$ gives the linear program
\begin{equation}\tag{$\mathsf{F}_{\mu,r}$}\label{eq:cap-LP}
\begin{aligned}
\min\quad &\rho\\
\text{s.t.}\quad
&\sum_{j\in\mu(i)}z_j+r_it_i+\sum_js_{ij}\le\rho B_i &&\forall i,\\
&s_{ij}\ge a_{ij}-z_j-t_i &&\forall i,j,\\
&s_{ij}\ge0,\quad z_j\ge0,\quad t_i\ge0,\quad\rho\ge0.
\end{aligned}
\end{equation}
Thus the optimal core factor of a fixed capacity matching is also exactly computable in polynomial time.

\end{document}